\documentclass[lettersize,journal]{IEEEtran}
\usepackage{algorithmic}
\usepackage{algorithm}
\usepackage{array}
\usepackage{soul}
\usepackage{amsmath,amsfonts,amssymb}
\usepackage{amsthm}
\usepackage[caption=false,font=normalsize,labelfont=sf,textfont=sf]{subfig}
\usepackage{textcomp}
\usepackage{stfloats}
\usepackage{url}
\usepackage{verbatim}
\usepackage{graphicx}
\usepackage[style=ieee,backend=bibtex]{biblatex}
\DeclareFieldFormat{doi}{} 
\usepackage{cleveref}
\crefname{figure}{Fig.}{Figs.}
\Crefname{figure}{Fig.}{Figs.}
\crefname{equation}{Eq.}{Eqs.}
\Crefname{equation}{Eq.}{Eqs.}
\usepackage{xcolor}
\crefname{figure}{Fig.}{Figs.}
\Crefname{figure}{Fig.}{Figs.}
\crefname{equation}{Eq.}{Eqs.}
\Crefname{equation}{Eq.}{Eqs.}

\newtheorem{lemma}{Lemma}

\Crefname{proposition}{Proposition}{Propositions}

\Crefname{corollary}{Corollary}{Corollaries}
\Crefname{lemma}{Lemma}{Lemmas}
\def\BibTeX{{\rm B\kern-.05em{\sc i\kern-.025em b}\kern-.08em
    T\kern-.1667em\lower.7ex\hbox{E}\kern-.125emX}}
\DeclareFieldFormat{doi}{} 

\usepackage{titlesec}
\titlespacing*{\section}{0pt}{1.05ex plus 0.2ex minus 0.2ex}{0.5ex plus 0.2ex}
\titlespacing*{\subsection}{0pt}{0.85ex plus 0.2ex minus 0.2ex}{0.4ex plus 0.1ex}
\titlespacing*{\subsubsection}{0pt}{0.75ex plus 0.2ex minus 0.2ex}{0.35ex plus 0.1ex}
\usepackage{indentfirst}
\newcommand{\tightsubsection}[1]{%
  \subsection{#1}\vspace{-0.4\baselineskip}%
}

\begin{document}

\title{GP Bandit-Assisted Two-Stage Sparse Phase Retrieval for Amplitude-Only Near-Field Beam Training}
\author{
Zijun Wang, Shawn Tsai, Ye Hu,~\IEEEmembership{Member,~IEEE}, Rui Zhang,~\IEEEmembership{Member,~IEEE}
\thanks{Zijun Wang, and Rui Zhang are with the Department of Electrical Engineering, The State University of New York at Buffalo, New York, USA  (email: {zwang267@buffalo.edu},{rzhang45@buffalo.edu}). Shawn Tsai is with the CSD, MediaTek Inc. USA, San Diego, CA 92122 USA (e-mail: shawn.tsai@mediatek.com). Ye Hu is with the Industrial and Systems Engineering, and Department of Electrical and Computer Engineering, University of Miami, Florida, USA (email: {yehu@miami.edu}). }
\thanks{This work was supported in part by the CSD, MediaTek Inc. USA, under Grant 103764, and in part by the National Science Foundation, under Grant ECCS 2512911.}
\thanks{Corresponding author: Rui Zhang.}
}

\markboth{Journal of \LaTeX\ Class Files,~Vol.~14, No.~8, August~2021}%
{Shell \MakeLowercase{\textit{et al.}}: A Sample Article Using IEEEtran.cls for IEEE Journals}

\IEEEpubid{0000-0000/00\$00.00~\copyright~2021 IEEE}

\maketitle

\begin{abstract}
The transition to Extremely Large Antenna Arrays (ELAA) in 6G introduces significant near-field effects, necessitating robust near-field beam training strategies in multi-path environments. Because signal phases are frequently compromised by hardware impairments such as phase noise and frequency offsets, amplitude-only channel recovery is a critical alternative to coherent beam training. However, existing near-field amplitude-based training methods often assume simplistic line-of-sight conditions. Conversely, far-field phase retrieval (PR) methods lack the sensing flexibility required to optimize training efficiency and are fundamentally limited by plane-wave models, making them ill-suited for near-field propagation.
We propose a two-stage sparse PR framework for amplitude-only near-field beam training in multipath channels. Stage I performs adaptive support discovery on the standard 2D DFT beamspace by exploiting a physics-guided prior induced by near-field beam patterns. Stage II then refines the channel estimate by restricting sensing and sparse PR to the learned subspace.
Numerical results show that the proposed adaptive pipeline consistently outperforms non-adaptive baselines, improving beamforming gain by over $70\%$ at low SNR.

\end{abstract}


\section{Introduction}
 \IEEEPARstart{T}{he} vision for 6G wireless networks centers on a multi-tier spectrum strategy, leveraging the upper mid-band, millimeter-wave, and Terahertz frequencies to meet the demand for Gbps-scale capacity, sub-millisecond latency, and ubiquitous coverage \cite{testolina_uppermidband_mcom24,10438977,9766110}. While these bands offer vast bandwidth, they also suffer from severe propagation and atmospheric absorption losses. To sustain a viable link budget, 6G base stations will deploy extremely large antenna arrays (ELAA) with hundreds or thousands of elements, providing the high beamforming gain necessary to overcome these path losses\cite{kang_uppermidband_ojcoms24, zhang_newmidband_mcom25}. However, the combination of massive physical apertures and short wavelengths pushes the some typical communication distances into the  near-field regime \cite{lu_nearfield_xlmimo_comst24,10942861}. In this regime, the traditional far-field plane-wave assumption is no longer valid. Instead, the channel is characterized by spherical wavefronts, where the array response depends on both the user's angular direction and their specific propagation range \cite{General_tutorial}.

This shift has motivated extensive studies on near-field channel modeling, channel estimation and beam management
\cite{Cui_and_Dai_channel_model,beam_split_wideband_beam_training,sequence_low_overhead_orth,k-b_domain_codebook}. { Under the assumption of phase-coherent measurements, near-field channel acquisition generally falls into two categories: parametric inference and compressive sparse recovery. Within the parametric framework, subspace-based methods leverage the eigenstructure of the sample covariance matrix to distinguish signal and noise subspaces, subsequently identifying path parameters through high-resolution spectrum searches. For instance, \cite{DFT_MLE_MUSIC} utilizes DFT-codebook sweeping to collect complex measurements, employing a MUSIC-type spectrum to isolate dominant angular components, followed by a matched-filter-based range scan for distance refinement. To mitigate the computational intensity of such searches, \cite{Zhang_nearfield_MUSIC} introduces a reduced-dimension near-field MUSIC formulation, enabling efficient joint angle-range estimation that scales effectively with the massive apertures of ELAAs. Complementing these subspace approaches are likelihood-driven estimation methods, which formulate channel acquisition as a maximum-likelihood estimation (MLE) problem under a spherical-wave model. Building on the polar-domain codebook proposed in \cite{Cui_and_Dai_channel_model}, a Gaussian MLE approach is used to refine off-grid Uniform Linear Array (ULA) estimates, while \cite{Scalable_Near-Field_Localization_UPA} extends this refinement to Uniform Planar Arrays (UPA). Beyond parametric fitting, the inherent sparsity of high-frequency channels allows near-field training to be framed as a compressive sensing (CS) problem. Early works in hybrid-beamforming architectures, such as \cite{lee_tcomm16_omp_hybrid_ce}, designed training precoders and combiners that utilize Orthogonal Matching Pursuit (OMP) to identify dominant spatial paths. This framework was further generalized in \cite{mendezrial_access16_hybrid_arch_cs_ce} to accommodate various hybrid hardware constraints. Transitioning to the near-field, \cite{wang2025sparsityawarenearfieldbeamtraining} exploits a polar-domain dictionary to represent the channel's joint angle-range sparsity, significantly reducing the reporting overhead of beam training compared to exhaustive two-dimensional sweeping under multi-path scenarios. }

\IEEEpubidadjcol
{ Despite the efficacy of coherent methods, phase stability is often a fragile assumption in practical transceivers. Residual synchronization errors and hardware impairments, including carrier-frequency offset, phase noise, and frequency drift, introduce unknown, time-varying distortions across training slots, preventing the coherent combination of pilot observations \cite{hu2020pccpr_tsp}. Consequently, there is a compelling need for non-coherent beam training and channel estimation strategies that rely exclusively on amplitude-only feedback. Under amplitude-only acquisition, the most intuitive baseline is codebook sweeping combined with power-based best beam selection. In near-field systems, polar-domain codebooks discretize both angle and range to provide the fine spatial focusing required for spherical-wave propagation \cite{Cui_and_Dai_channel_model}. However, exhaustive searching over a two-dimensional angle-range grid incurs prohibitive overhead. To mitigate this, hierarchical and structured codebook designs have been proposed to localize promising spatial regions before performing fine-grained refinement. For instance, \cite{Hierarchical_near-field_beam_relocate} utilizes a two-layer polar-domain codebook for low-overhead region searching, while \cite{k-b_domain_codebook} introduces a "slope-intercept" representation to decouple chirp-like parameters, organizing beams into a hierarchical search that reduces inter-codeword correlation. While effective for single-path Line-of-Sight (LoS) scenarios, these sweeping strategies often overlook multi-path environments where multiple spatial components must be resolved.}

{ In multi-path channels, robust beamforming requires recovering both the dominant paths and their associated complex combining coefficients. With amplitude-only measurements, this task evolves into a phaseless inverse problem, where the goal is to recover the channel vector $\mathbf{x}$ from magnitude measurements $y_i = |\langle \mathbf{a}_i,\mathbf{x}\rangle|$ or $y_i^2 = |\langle \mathbf{a}_i,\mathbf{x}\rangle|^2$. Classical Wirtinger Flow (WF) solvers, such as Truncated Wirtinger Flow (TWF) and
Truncated Amplitude Flow (TAF), employ spectral initialization followed by gradient-descent updates on intensity-based losses, using truncation rules to enhance robustness against outliers \cite{candes2015wf_tit, wang2018taf_tit}. However, these general solvers do not exploit the structured sparsity inherent in high-frequency wireless channels. While sparse PR variants like Sparse Wirtinger Flow (SWF) and Sparse Phase Retrieval via Truncated Amplitude Flow (SPARTA) \cite{yuan2019swf, wang2018sparta_tsp} incorporate hard-thresholding to pursue sparse supports, they are not tailored for beam training; they typically require a pre-specified sparsity level and rely on non-adaptive sensing matrices, missing the opportunity to reduce sweeping overhead through intelligent, real-time probing. Non-coherent beam training has been studied in far-field wireless systems, which is typically framed as sparse PR or phaseless decoding. Works such as \cite{li2019fast_tsp} enable the recovery of dominant angular beams from compressive magnitude measurements, while \cite{hu2020pccpr_tsp} and \cite{Qiu_TSP_far-field_phaseretrieval} handle unknown phase offsets and channel tracking through phase-modulated probing. However, these methods are built upon the plane-wave far-field model, where sparsity is defined on an angular dictionary. They fail to account for the near-field spherical-wave effect, where steering depends jointly on angle and range.}

{ Consequently, a significant gap remains: how to achieve low-overhead, non-coherent near-field beam training in multi-path environments? To address this challenge, we propose a Gaussian bandit-assisted sparse phase retrieval framework for amplitude-only near-field beam training. Our approach utilizes a two-stage design. First, the base station performs a coarse, low-overhead support discovery using DFT beams, guided by an adaptive Gaussian Process (GP) bandit policy to prioritize promising directions. Second, the base station executes sparse PR exclusively on the estimated support using a Gaussian-masked DFT sensing matrix. To the best of our knowledge, this is the first work to perform amplitude-only near-field beam training for multi-path channels with an adaptive, bandit-guided sparse phase retrieval framework.
}

The main contributions of this work are summarized as follows:
\begin{itemize}
    \item \textbf{DFT beamspace modeling for near-field channels and beam pattern analysis}: { We adopt the standard DFT codebook as the beamspace representation to ensure compatibility with far-field systems. To address the inherent near-field spatial mismatch, we provide a rigorous analysis of the DFT beam pattern under the spherical-wave effect. Specifically, we derive a closed-form expression for the lobe beamwidth under energy-splitting, revealing a structured and predictable spatial distribution that governs near-field channels in the DFT domain.}
 \item \textbf{Physically-informed kernel design and adaptive support discovery}: 
Leveraging our derived beamwidth analysis, we design a physically-motivated kernel for a Gaussian Process bandit defined over the DFT beam indices. This bandit-based policy enables an intelligent, low-overhead search for the dominant channel support (i.e., significant paths) using only non-coherent, amplitude-only measurements, significantly outperforming other baselines. 
    \item \textbf {Rician-aware amplitude-only recovery}:  We tailor the Gaussian bandit updates and the sparse phase retrieval refinement to the specific statistics of beam training. By incorporating a Rician denoising step, our framework explicitly accounts for the distribution of amplitude-only feedback in the presence of complex Gaussian noise. This adaptation yields superior reconstruction accuracy and robustness compared to conventional PR solvers that assume simplistic Gaussian noise models. Moreover, by coupling high-probability sub-Gaussian noise proxy induced by our Rician debiasing with the designed kernel prior, the Stage~I procedure admits an agnostic GP-bandit (RKHS) justification \cite{srinivas2010gpucb}.
    
\end{itemize}

The remainder of this paper is organized as follows. Section \ref{sec:system_model} introduces the near-field signal model and the amplitude-only measurement formulation. Section \ref{sec:upa_beam_pattern} analyzes the near-field channel beam pattern in the DFT space, facilitating the design of our specialized kernel. Section \ref{sec:proposed} presents the proposed framework, detailing the bandit-assisted support discovery and the sparse PR refinement with Gaussian-masked DFT sensing. Section \ref{sec:simulation} provides comprehensive simulation results, and Section \ref{sec:clu} concludes the paper.

\noindent\textit{Notation:}
Scalars, vectors, and matrices are denoted by $a$, $\mathbf{a}$, and $\mathbf{A}$.
$(\cdot)^{\mathsf T}$ and $(\cdot)^{\mathsf H}$ denote transpose and conjugate transpose.
$\|\cdot\|_2$ is the Euclidean norm and $\|\cdot\|_0$ counts nonzero entries.
$\otimes$ denotes the Kronecker product and $\Re\{\cdot\}$ denotes the real part.
$[x]_+\triangleq \max\{x,0\}$ and $\mathcal{H}_k(\cdot)$ keeps the $k$ largest-magnitude entries.
$|\mathcal{S}|$ denotes set cardinality and $\mathcal{S}^c$ denotes set complement.
$\mathcal{CN}(\mu,\Sigma)$ and $\mathcal{U}[a,b]$ denote complex Gaussian and uniform distributions, respectively.

\section{System Model}\label{sec:system_model}
\noindent\hspace*{1em}
We study a downlink narrowband scenario where a base station (BS) equipped with an
$N_y\times N_z$ UPA serves a single-antenna user.
The UPA is placed on the $y$-$z$ plane and centered at the origin.
Let $f_c$ denote the carrier frequency and $\lambda=c/f_c$ the wavelength, where $c$ is the
speed of light.
The considered geometry is illustrated in \Cref{fig:antenna}.

\begin{figure}[h]
\captionsetup{justification=justified,singlelinecheck=false} 
\centering
\includegraphics[width=0.8\linewidth]{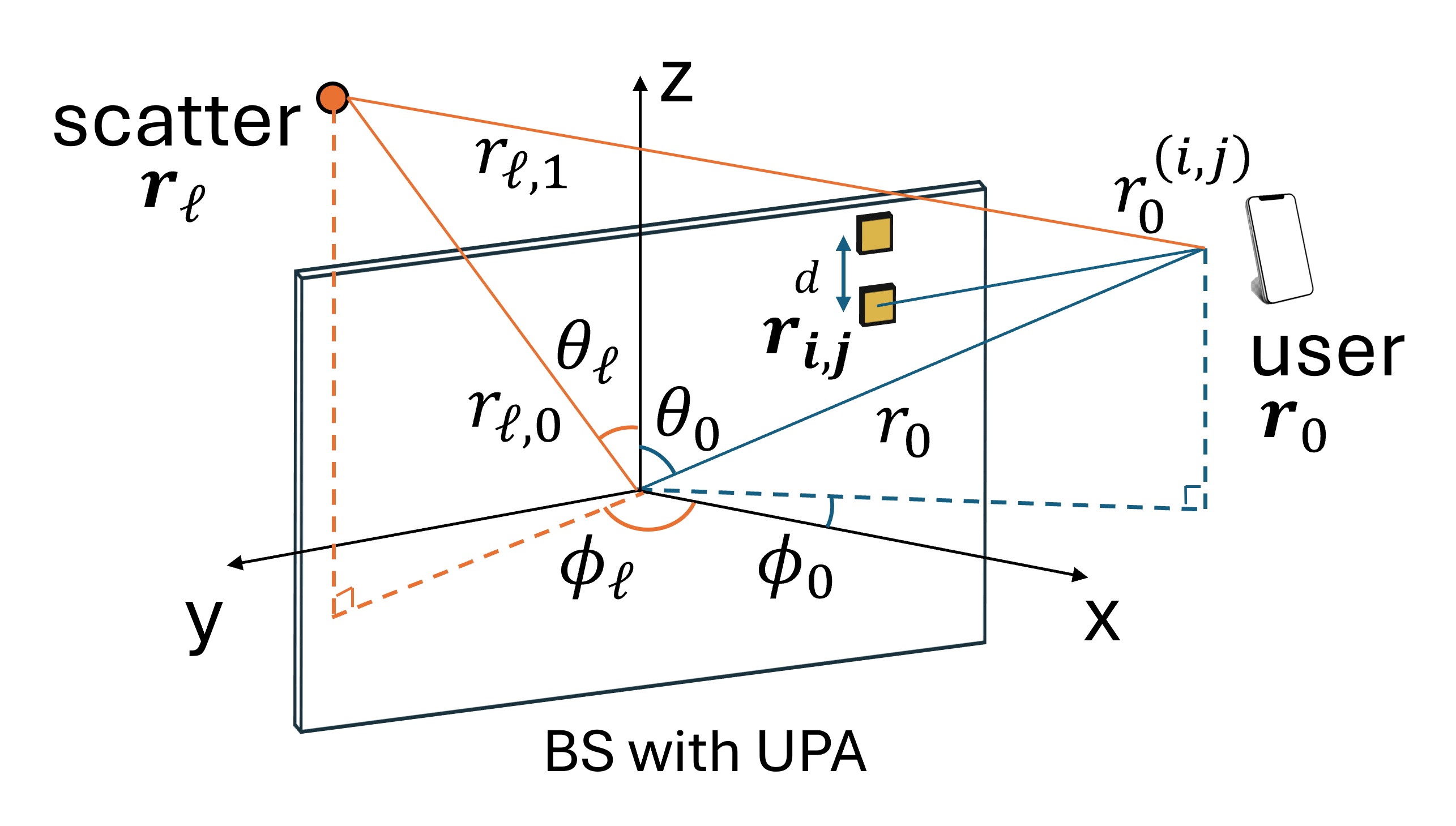}
\caption{Diagram of the near-field communication system featuring a UPA.}
\label{fig:antenna}
\end{figure}

Let
$
\mathcal{N}_y \triangleq \{0,1,\dots,N_y-1\},\;
\mathcal{N}_z \triangleq \{0,1,\dots,N_z-1\},
$
and define the centered index offsets
$
\delta_i \triangleq \frac{2i-N_y+1}{2},\;
\delta_j \triangleq \frac{2j-N_z+1}{2},
\; i\in\mathcal{N}_y,\; j\in\mathcal{N}_z.
$
The $(i,j)$-th antenna element is located at
$
\mathbf{r}_{i,j} = (0,\,\delta_i d,\,\delta_j d)^{\mathsf{T}},
$
where $d=\lambda/2$ is the inter-element spacing and $N\triangleq N_yN_z$ is the total number of antennas.

A point in space (user or scatterer) is parameterized by spherical coordinates:
range $r>0$, elevation $\theta\in[0,\pi]$, and azimuth $\phi\in[-\pi,\pi]$.
We use the unit direction vector
\begin{equation}
\mathbf{q}(\theta,\phi) \triangleq
\begin{bmatrix}
\sin\theta\cos\phi\\[0.5mm]
\sin\theta\sin\phi\\[0.5mm]
\cos\theta
\end{bmatrix},
\end{equation}
so that the corresponding Cartesian position is $r\,\mathbf{q}(\theta,\phi)$.
The distance from this point to the $(i,j)$-th array element is
$
r^{(i,j)}(\theta,\phi,r)
= \big\|r\,\mathbf{q}(\theta,\phi)-\mathbf{r}_{i,j}\big\|_2.
$
Accordingly, the near-field steering response associated with this path is
\begin{equation}
\big[b(\theta,\phi,r)\big]_{i,j}
= \frac{1}{\sqrt{N}}
\exp\!\left(
 -j\frac{2\pi}{\lambda}
 \big(r^{(i,j)}(\theta,\phi,r)-r\big)
\right),
\label{eq:upa_steering_entry_compact}
\end{equation}
and stacking all $(i,j)$ entries in a fixed order gives the steering vector
$\mathbf{b}(\theta,\phi,r)\in\mathbb{C}^{N}$.

We place the user at $(\theta_0,\phi_0,r_0)$, where $r_0$ is the BS-user distance.
The line-of-sight (LoS) component is modeled as
\begin{equation}
\mathbf{h}_0
= \sqrt{N}\,g_0
  e^{-j\frac{2\pi r_0}{\lambda}}
  \mathbf{b}(\theta_0,\phi_0,r_0),
\label{eq:upa_channel_los_r0}
\end{equation}
with the large-scale gain $g_0 = \frac{\lambda}{4\pi r_0}$.

In addition to LoS, we consider $(L-1)$ single-bounce scatterers indexed by $\ell=1,\dots,L-1$.
The $\ell$-th scatterer is located at $(\theta_\ell,\phi_\ell,r_{\ell,0})$ relative to the BS, with position
$
\mathbf{r}_\ell \triangleq r_{\ell,0}\,\mathbf{q}(\theta_\ell,\phi_\ell).
$
The BS-scatter distance is $r_{\ell,0}=\|\mathbf{r}_\ell\|_2$, and the scatter-user distance is
$
r_{\ell,1}
= \big\|r_0\,\mathbf{q}(\theta_0,\phi_0)
      - r_{\ell,0}\,\mathbf{q}(\theta_\ell,\phi_\ell)\big\|_2.
$
The resulting NLoS contribution is
\begin{equation}
\mathbf{h}_\ell
= \sqrt{N}\,g_\ell
  e^{-j\frac{2\pi (r_{\ell,0}+r_{\ell,1})}{\lambda}}
  \mathbf{b}(\theta_\ell,\phi_\ell,r_{\ell,0}),
\label{eq:upa_channel_each_nlos_r_l0_l1}
\end{equation}
where
$g_\ell
= \frac{\lambda}{4\pi r_{\ell,0} r_{\ell,1}}\,p_\ell$ and $p_\ell \sim \mathcal{CN}(0,1)$
capture the two-hop path loss and the random complex reflection coefficient, respectively.
Collecting all paths, the downlink channel is
\begin{equation}
\mathbf{h}
= \mathbf{h}_0 + \sum_{\ell=1}^{L-1}\mathbf{h}_\ell
\in\mathbb{C}^{N}.
\label{eq:upa_channel_total_compact}
\end{equation}

In this work, our focus is the  near-field regime, in which the user range falls between the Fresnel
distance $r_{\mathrm{F}}$ and the Rayleigh distance $r_{\mathrm{R}}$ \cite{wuMultipleAccessNearField2023}.  Here $r_{\mathrm{F}}\triangleq 0.62\sqrt{D^3/\lambda}$ and
$r_{\mathrm{R}}\triangleq 2D^2/\lambda$, where $D$ is the effective aperture size of the UPA.

\section{Beam Pattern and Sparsity on UPA}
\label{sec:upa_beam_pattern}

\noindent\hspace*{1em}
This section characterizes how a near-field path appears in the 2D DFT beamspace of a UPA.
In far-field beamspace, the energy of a single path is typically concentrated around one DFT index.
In the near field, however, the Fresnel (quadratic) phase induces energy spreading across neighboring DFT indices.
More importantly for our design, this spreading exhibits strong axis-aligned locality: due to the separable structure of the Fresnel phase on a UPA, the beamspace power map becomes highly correlated along the two DFT axes, forming a cross-shaped local correlation pattern around the dominant beam.
To make this behavior explicit, we first introduce a separable Fresnel approximation under which the UPA response factorizes into two 1D patterns, and then use the 6-dB lobe-width law for near-field ULAs in \cite{wang2025lowcomplexitynearfieldbeamtraining} to quantify the per-axis spreading widths and the resulting effective sparsity.

We parameterize the 2D DFT beamspace using
$
u \triangleq \sin\theta\,\sin\phi,
\;
v \triangleq \cos\theta.
$
With this choice, the orthogonal 2D DFT sampling corresponds to uniform grids over $(u,v)$:
$
u_{n}=\frac{2n-N_y+1}{N_y},
\;
v_{m}=\frac{2m-N_z+1}{N_z},
$
for $n\in\{0,\dots,N_y-1\}$ and $m\in\{0,\dots,N_z-1\}$.

We use the standard DFT codebook as the beamspace basis on the UPA.
Along the $y$-axis, for any spatial cosine $u\in[-1,1]$, define the unit-norm 1D DFT vector
{\small\begin{equation}
\mathbf{a}_y(u)\triangleq
\frac{1}{\sqrt{N_y}}
\big[
e^{-j\frac{2\pi}{\lambda}du\delta_0},
e^{-j\frac{2\pi}{\lambda}du\delta_1},
\dots,
e^{-j\frac{2\pi}{\lambda}du\delta_{N_y-1}}
\big]^{\mathsf T},
\label{eq:1d_dft_ay}
\end{equation}}
and similarly along the $z$-axis
{\small\begin{equation}
\mathbf{a}_z(v)\triangleq
\frac{1}{\sqrt{N_z}}
\big[
e^{-j\frac{2\pi}{\lambda}dv\delta_0},
e^{-j\frac{2\pi}{\lambda}dv\delta_1},
\dots,
e^{-j\frac{2\pi}{\lambda}dv\delta_{N_z-1}}
\big]^{\mathsf T},
\label{eq:1d_dft_az}
\end{equation}}
A 2D DFT codeword pointing to $(u_n,v_m)$ is then
$
\mathbf{a}(u_n,v_m)
\triangleq
\mathbf{a}_y(u_n)\otimes \mathbf{a}_z(v_m)\in\mathbb{C}^{N_yN_z}.
$

\subsection{Separable Fresnel approximation}
\noindent\hspace*{1em}
Let $r^{(i,j)}$ be the distance from element $(i,j)$ to a path at $(\theta,\phi,r)$.
A second-order (Fresnel) expansion around the array center gives
{\small
\begin{align}
r^{(i,j)}-r
\approx
&-d\big(u\,\delta_i+v\,\delta_j\big)
+\frac{d^2}{2r}\Big((1-u^2)\delta_i^2+(1-v^2)\delta_j^2\Big)\nonumber\\
&-\frac{d^2}{r}\,u v\,\delta_i\delta_j.
\label{eq:fresnel_upa_full}
\end{align}}
Following standard large-array UPA analyses \cite{wuMultipleAccessNearField2023},
the mixed quadratic term $-\tfrac{d^2}{r}uv\,\delta_i\delta_j$ is typically much smaller than the
axis-wise quadratic terms over the operating near-field region; we therefore neglect it.
This yields an approximately separable near-field phase across the two axes, so the steering vector admits
\begin{equation}
\mathbf{b}(u,v,r) ~\approx~ \mathbf{b}_y(u,r)\otimes \mathbf{b}_z(v,r),
\label{eq:upa_separable_sv}
\end{equation}
where
{\footnotesize
\begin{align}
\big[\mathbf{b}_y(u,r)\big]_i
&=\frac{1}{\sqrt{N_y}}
\exp\!\Big\{-j\frac{2\pi}{\lambda}\Big(-d\,u\,\delta_i+\frac{d^2}{2r}(1-u^2)\delta_i^2\Big)\Big\},
\\
\big[\mathbf{b}_z(v,r)\big]_j
&=\frac{1}{\sqrt{N_z}}
\exp\!\Big\{-j\frac{2\pi}{\lambda}\Big(-d\,v\,\delta_j+\frac{d^2}{2r}(1-v^2)\delta_j^2\Big)\Big\}.
\label{eq:upa_separable_sv_uv}
\end{align}}

\subsection{2D beam-pattern factorization and plateau support}
\noindent\hspace*{1em}
Consider a path at $(u_0,v_0,r_0)$. Under \Cref{eq:upa_separable_sv}, the inner product between the near-field steering vector
and a 2D DFT codeword factorizes as
\begin{align}
&G(u_0,v_0,r_0;u_n,v_m)\nonumber\\
&\triangleq \big|\mathbf{b}(u_0,v_0,r_0)^{H}\mathbf{a}(u_n,v_m)\big|\nonumber\\
&\approx
\big|\mathbf{b}_y(u_0,r_0)^{H}\mathbf{a}_y(u_n)\big|\;
\big|\mathbf{b}_z(v_0,r_0)^{H}\mathbf{a}_z(v_m)\big|.
\label{eq:2d_factorization}
\end{align}
Hence, the 2D magnitude response can be understood as the product of two 1D responses, one per axis.

To quantify the 1D spreading, we normalize the axis-wise responses. For the $z$-axis, define
\begin{align}
G_{z,\mathrm{norm}}(v_0,r_0;v)
&\triangleq
\frac{\big|\mathbf{b}_z(v_0,r_0)^{H}\mathbf{a}_z(v)\big|}
{\big|\mathbf{b}_z(v_0,r_0)^{H}\mathbf{a}_z(v_0)\big|}.
\label{eq:axis_norm_resp}
\end{align}
The corresponding 6-dB lobe set is
\begin{align}
\Phi^{z}_{6\mathrm{dB}}(v_0,r_0)
&\triangleq
\bigl\{v\mid\,G_{z,\mathrm{norm}}(v_0,r_0;v)\ge 1/2\bigr\},
\end{align}
and its width is
\begin{align}
B_z(v_0,r_0)\triangleq \max(\Phi^{z}_{6\mathrm{dB}})-\min(\Phi^{z}_{6\mathrm{dB}}).
\label{eq:axis_width_def}
\end{align}
The same definitions apply to the $y$-axis.

For near-field ULAs, \cite{wang2025lowcomplexitynearfieldbeamtraining} shows that the 6-dB lobe width for a path at $(\xi_0,r_0)$ admits
\begin{equation}
B(\xi_0,r_0)=Nd\,\frac{1-\xi_0^2}{r_0}.
\label{eq:beamwidth_ula}
\end{equation}
In our UPA model, \Cref{eq:upa_separable_sv_uv} reveals that $\mathbf{b}_y(u_0,r_0)$ and $\mathbf{b}_z(v_0,r_0)$ share the same
quadratic-phase structure as the ULA steering vector in \cite{wang2025lowcomplexitynearfieldbeamtraining}, but with apertures $N_y d$ and $N_z d$.
Applying the same 6-dB width law along each axis yields
\begin{equation}
B_y(u_0,r_0)=N_y d\,\frac{1-u_0^2}{r_0},\;
B_z(v_0,r_0)=N_z d\,\frac{1-v_0^2}{r_0}.
\label{eq:beamwidth_upa_uv}
\end{equation}
These widths quantify the extent of near-field energy spreading around the peak in the continuous $(u,v)$ domain.
Combined with the factorization in \Cref{eq:2d_factorization}, a convenient support approximation is to take the Cartesian product
of the two axis-wise 6-dB sets, which aligns with the connected plateau with sharp edges observed in 2D DFT beamspace
\cite{wang2025lowcomplexitynearfieldbeamtraining,Cui_and_Dai_channel_model}.

\begin{figure}[!t]
\centering
\includegraphics[width=0.8\linewidth]{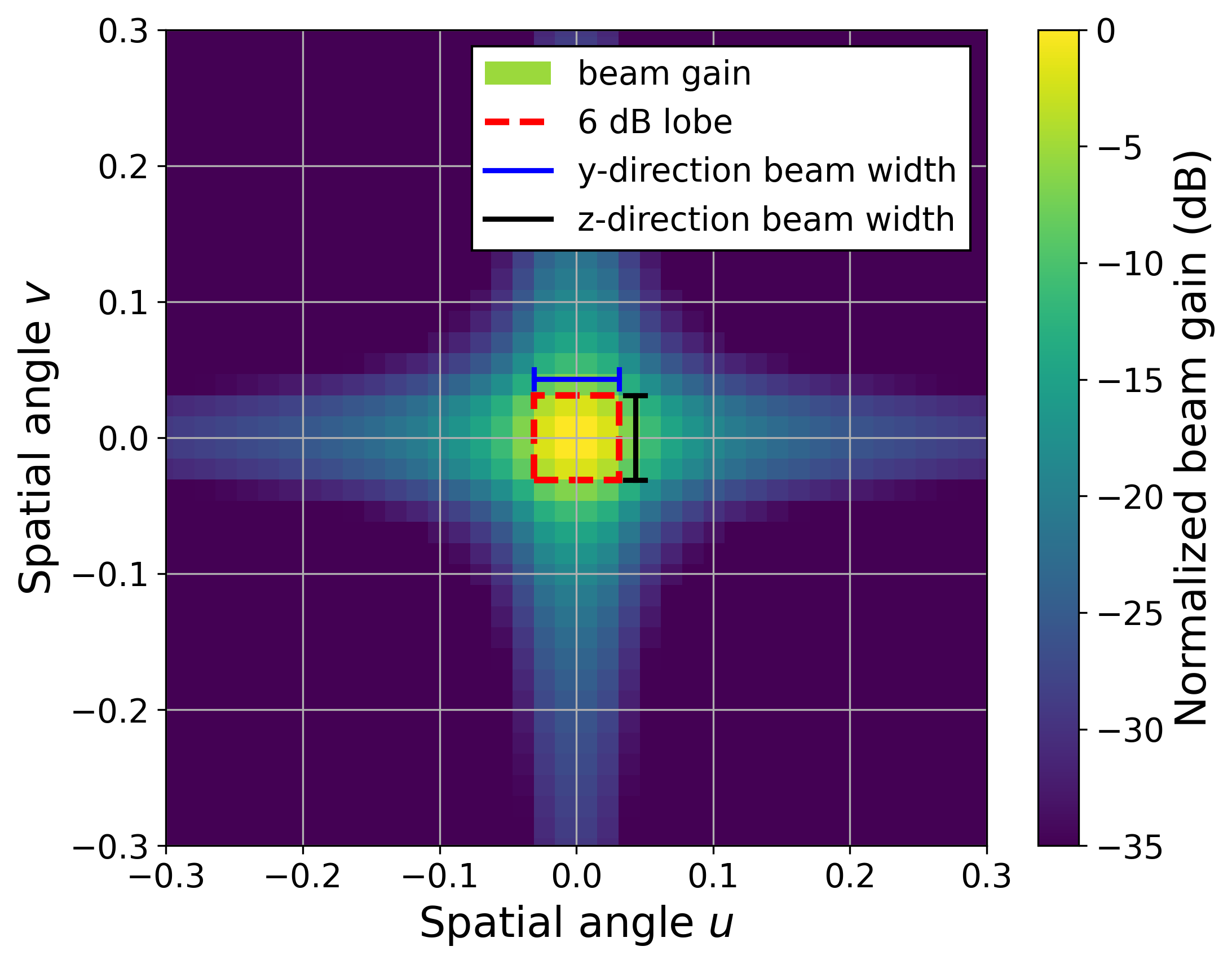}
\caption{Normalized DFT-beamspace magnitude on a $(128\times 128)$ UPA at
$f_c=28~\mathrm{GHz}$.}
\label{fig:upa_mainlobe}
\end{figure}

\begin{figure}[t]
    \centering
    \subfloat[\tiny(c)][\textrm{\footnotesize u-cut (y-direction beam width)}]{
        \includegraphics[width=0.46\linewidth]{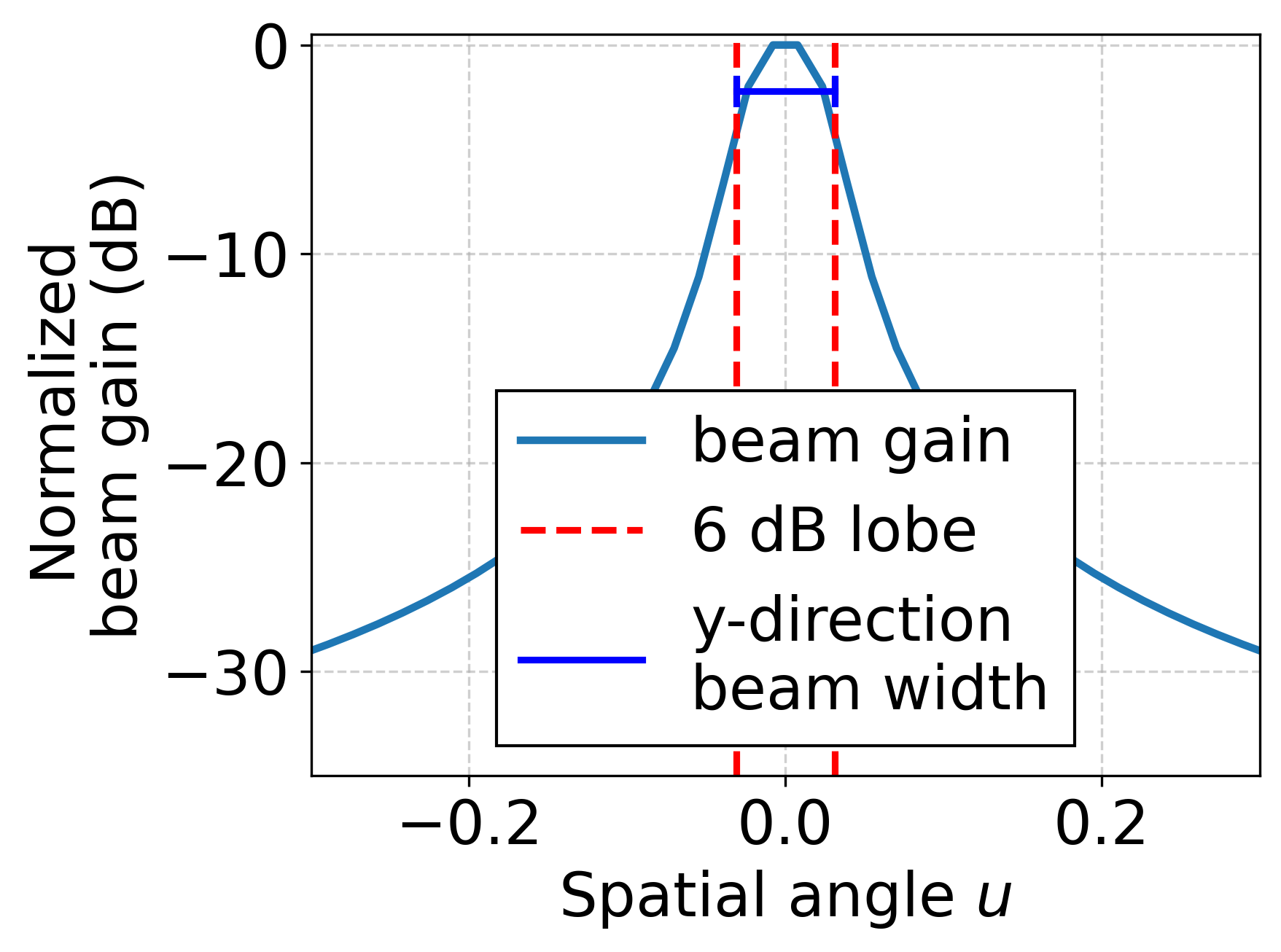}%
        \label{fig:upa_cut_u}
    }
    \hfill
    \subfloat[\tiny(c)][\textrm{\footnotesize v-cut (z-direction beam width)}]{
        \includegraphics[width=0.46\linewidth]{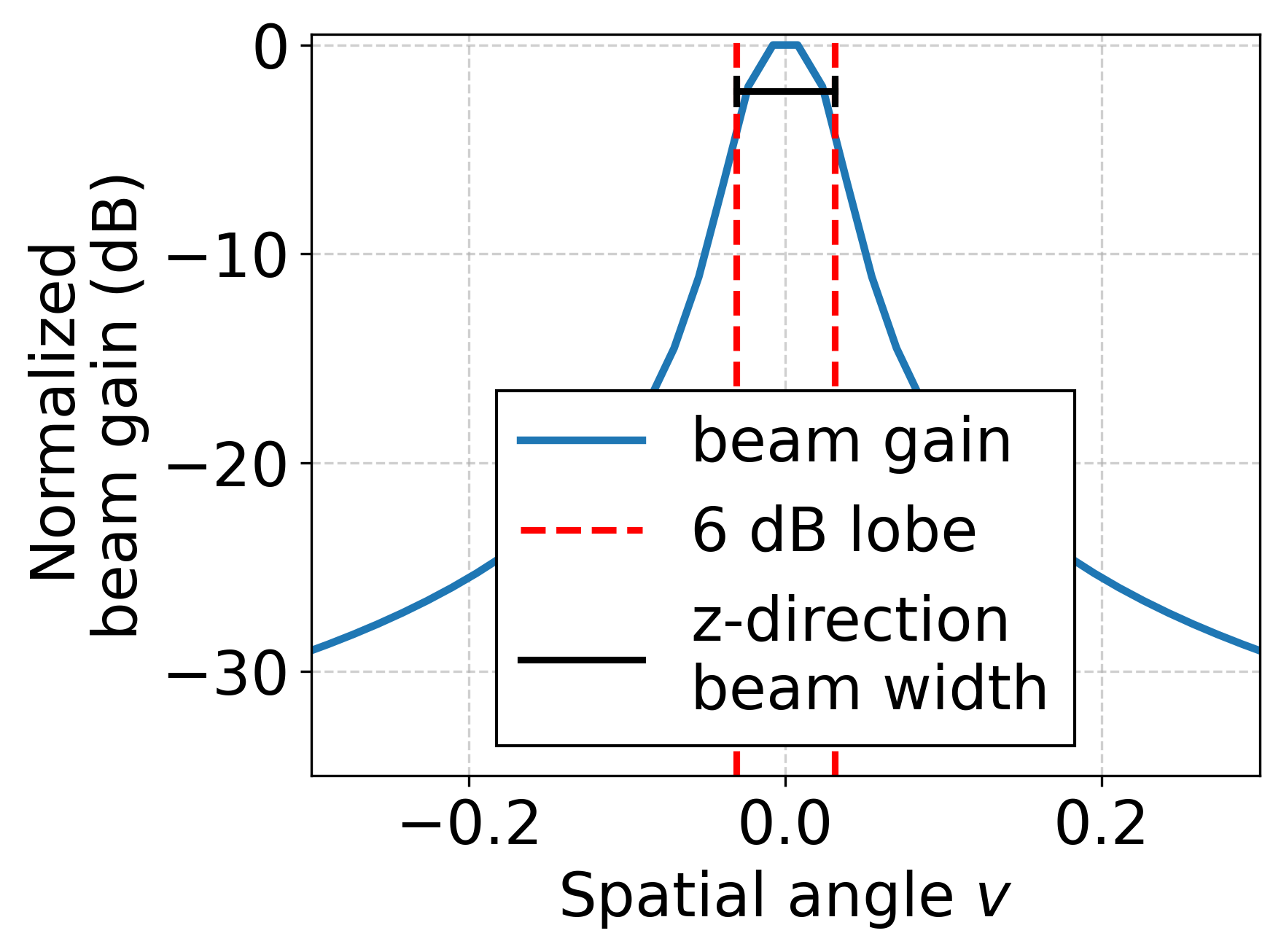}%
        \label{fig:upa_cut_v}
    }
    \caption{
        1D cuts of the normalized gain and the 6-dB lobe widths.
    }
    \label{fig:upa_main_lobe2}
\end{figure}
An example of the UPA beam pattern is shown in \Cref{fig:upa_mainlobe}.
The red dashed rectangle marks a 6-dB support region, which serves as a convenient approximation for the effective active neighborhood.
\Cref{fig:upa_main_lobe2} further reports the 1D $u$- and $v$-cuts through the peak of \Cref{fig:upa_mainlobe}; the 6-dB thresholds delineate the per-axis lobe widths $B_y$ and $B_z$.
Together with the factorization in \Cref{eq:2d_factorization}, these cuts highlight strong axis-aligned locality, i.e., nearby DFT indices along the $u$- and $v$-directions tend to have similar magnitudes, giving rise to the cross-shaped correlation pattern.

\subsection{Effective sparsity level in 2D DFT beamspace}
\noindent\hspace*{1em}Let $\Delta_y$ and $\Delta_z$ denote the DFT grid spacings along the
$u$- and $v$-axes, respectively. Approximating each path's 2D
lobe support by a rectangle with side lengths $B_y$ and $B_z$, the
number of active DFT coefficients per path is approximately
$B_yB_z/(\Delta_y\Delta_z)$. For $L$ paths with non-overlapping lobe
supports (the worst case), the expected sparsity level $K_{\mathrm{UPA}}$ (i.e., the expected number of active
DFT coefficients) scales as
\begin{equation}
\mathbb{E}[K_{\mathrm{UPA}}]
~\approx~
\frac{LN_yN_zd^2}{\Delta_y\Delta_z}\,
\mathbb{E}\!\left[\frac{(1-u^2)(1-v^2)}{r^2}\right].
\label{eq:sparsity_upa_general}
\end{equation}
Note that $u$ and $v$ are coupled through
$u=\sin\theta\sin\phi=\sqrt{1-v^2}\,s$ with $s\triangleq\sin\phi$. We assume the following independent uniform distributions:
\begin{equation}
v=\cos\theta \sim \mathcal{U}[v_1,v_2],\;
s=\sin\phi \sim \mathcal{U}[s_1,s_2],\;
r\sim\mathcal{U}[r_1,r_2],
\label{eq:dist_assump_us_r}
\end{equation}
and $v,s,r$ are independent to each other. Then the expectation in \Cref{eq:sparsity_upa_general} admits a closed form, and the sparsity
approximation becomes
\begin{equation}
\mathbb{E}[K_{\mathrm{UPA}}]
~\approx~
\frac{LN_yN_zd^2}{\Delta_y\Delta_z}\,
\Big(\Xi_{v,s}\Big)\,
\Big(\Xi_{r}\Big),
\label{eq:sparsity_upa_closed}
\end{equation}
where
\begin{align}
\Xi_{r}
=&\frac{1}{r_2-r_1}\left(\frac{1}{r_1}-\frac{1}{r_2}\right),
\label{eq:Xi_r}
\\
\Xi_{v,s}
=&\Big(1-\mu_2\Big)
-\nu_2\Big(1-2\mu_2+\mu_4\Big),
\label{eq:Xi_us}
\end{align}
with the moments
\begin{align}
&\mu_2=\frac{v_2^3-v_1^3}{3(v_2-v_1)},\;
\mu_4=\frac{v_2^5-v_1^5}{5(v_2-v_1)},\nonumber\\
&\nu_2=\frac{s_2^3-s_1^3}{3(s_2-s_1)}.
\label{eq:uv_moments}
\end{align}
A detailed derivation of \Cref{eq:sparsity_upa_closed} is provided
in \Cref{app:sparsity_upa_derivation}. For a $N_y\times N_z=128\times 16$ UPA at $f_c=28$~GHz with $d=\lambda/2$
and the standard 2D DFT grid $\Delta_y=2/N_y$ and $\Delta_z=2/N_z$, consider
$
v\sim \mathcal{U}\!\left[-\tfrac{1}{2},\tfrac{1}{2}\right],\;
s\sim \mathcal{U}\!\left[-\tfrac{1}{2},\tfrac{1}{2}\right],\;
r\sim \mathcal{U}\!\left[r_{\mathrm{F}},\frac{r_{\mathrm{R}}}{20}\right].
$
In this case, for an $L=6$-path channel, \Cref{eq:sparsity_upa_general} yields
$
\mathbb{E}[K_{\mathrm{UPA}}]\approx 11,
$
confirming that the near-field channel remains strongly sparse in the 2D DFT beamspace under this setting.
Moreover, the preceding analysis implies pronounced cross-shaped local correlation of DFT-beam magnitudes around only a few dominant components.
In the sequel, we exploit this structure to design a cross-pattern kernel for Stage~I search and to justify sparse recovery over the standard 2D DFT codebook for Stage~II PR.

\section{Proposed Gaussian Bandit-Assisted Sparse Phase Retrieval}
\label{sec:proposed}
This section presents a two-stage amplitude-only near-field beam training framework for UPAs, which operates in the DFT beamspace.
Stage~I performs adaptive support discovery by sequentially probing a small subset of 2D DFT beams and using a GP bandit to learn the beamspace power over the DFT grid. The key idea is to exploit the structured near-field DFT beam pattern characterized in \Cref{sec:upa_beam_pattern}: due to near-field energy splitting, neighboring DFT indices exhibit locally correlated magnitudes within a compact lobe region, which motivates a smooth GP prior and enables an efficient adaptive search.
Stage~II then refines the channel estimate by performing sparse phase retrieval over the discovered support set in Stage~I using a Gaussian-masked DFT sensing matrix. By restricting the search space, Stage~II operates on a much lower ambient dimension and thus requires fewer amplitude measurements than a full-grid recovery, while still being able to further prune and refine the support locally.


\subsection{Stage I: GP-LSE Support Discovery in 2D DFT Beamspace}
\label{sec:stage1_gpucb}

Let $\mathcal{I}\triangleq\{(n,m):n\in\{0,\dots,N_y-1\},\,m\in\{0,\dots,N_z-1\}\}$ denote the 2D DFT beam index set. 
For each $i=(n,m)\in\mathcal{I}$, let $\mathbf{f}_i\in\mathbb{C}^{N}$ denote the corresponding unit-norm 2D DFT codeword on the UPA (i.e., the beam pointing to the DFT grid point $(u_n,v_m)$ introduced in Section~\ref{sec:upa_beam_pattern}). 
{ During Stage~I, the BS restricts its probing beams to this DFT codebook. Specifically, at training slot $t$, the BS transmits a known pilot $x_t$ (with $|x_t|=1$) using the $i_t$-th DFT beamforming vector} $\mathbf{v}_{i_t} = \mathbf{f}_{i_t}, \; i_t=(n_t,m_t)\in\mathcal{I}.$
The received signal at the user is
\begin{equation}
y_t = \mathbf{h}^{H}\mathbf{v}_{i_t} x_t + w_t,\;
w_t\sim\mathcal{CN}(0,\sigma^2).
\label{eq:stage1_rx}
\end{equation}
Due to phase impairments, we use the power feedback $z_t^{\mathrm{raw}}=|y_t|^2$ and define the beamspace power map
$f(i)\triangleq|\mathbf{h}^H\mathbf{f}_i|^2$. Since $\mathbb{E}[z_t^{\mathrm{raw}}\,|\,\mathbf{h}]=f(i_t)+\sigma^2$,
we subtract the noise bias and obtain
\begin{equation}\label{eq:stage1_power_denoised}
\tilde z_t \triangleq z_t^{\mathrm{raw}}-\sigma^2
= f(i_t)+\tilde\varepsilon_t,
\end{equation}
where $\tilde\varepsilon_t$ is conditionally zero-mean and captures the fluctuation of the noncentral
Rician power around its mean. For tractable GP regression, we compute $(\mu_t(\cdot),\sigma_t(\cdot))$ using the standard
Gaussian-likelihood GP update with regularization parameter $\sigma_\varepsilon^2$ in \Cref{eq:Ktilde}.
Notably, the debiased power noise $\tilde\varepsilon_t$ in \Cref{eq:stage1_power_denoised} is generally
non-Gaussian. To enable high-probability confidence statements,
we construct in Appendix~\ref{app:stat_stage1} a bounded-noise event $\mathcal{E}_{\mathrm{bd}}(T_1)$ such that,
with probability at least $1-\delta_{\mathrm{bd}}$, the feedback noise forms a uniformly bounded martingale
difference sequence over the Stage~I horizon:
$|\tilde\varepsilon_t|\le B_t\le B_{T_1}$ for all $t\le T_1$. This bounded-noise model is then used in
Appendix~\ref{app:lse_overhead} to derive the GP-LSE guarantee via Theorem~6 of \cite{srinivas2010gpucb}.

We place a GP prior on the power map $f(\cdot)$:
\begin{equation}
f(\cdot)\sim\mathcal{GP}\!\big(0,\,k(\cdot,\cdot)\big).
\label{eq:gp_prior}
\end{equation}

Motivated by the empirical near-field DFT beam pattern in \Cref{fig:upa_mainlobe}, the beamspace power map inherits a compact plateau around the peak and a cross-shaped correlation along the two DFT axes. In particular, if a given DFT beam exhibits a large received power, then its neighboring beams along the $u$- and $v$-directions are also likely to be strong due to near-field energy spreading. To encode this physics-driven coupling, so that each probed beam can inform a cross-shaped neighborhood of unprobed beams, we adopt a \emph{cross-pattern} kernel defined on the 2D DFT coordinates.

Associate each index $i=(n,m)$.
We define the 1D Laplacian kernels
{\small\begin{equation}
k_u(i,i')=\exp\!\left(-\frac{|u_n-u_{n'}|}{\ell_u}\right),
k_v(i,i')=\exp\!\left(-\frac{|v_m-v_{m'}|}{\ell_v}\right),
\label{eq:laplace_1d_kernel}
\end{equation}}
and construct a \emph{cross-pattern} combination
\begin{equation}
k(i,i')
=
\alpha\,\frac{k_u(i,i')+k_v(i,i')}{2}
+(1-\alpha)\,k_u(i,i')\,k_v(i,i'),
\label{eq:plus_laplace_kernel}
\end{equation}
where $\alpha\in[0,1]$ trades off axis-aligned (cross-shaped) correlation vs.\ joint 2D correlation, and $\ell_u,\ell_v$ control the correlation span along the $y$- and $z$-axes, respectively. 
To reflect the near-field beamwidth law \Cref{eq:beamwidth_upa_uv}, a physics-guided choice links $(\ell_u,\ell_v)$ to typical lobe widths,
\begin{align}
   &\ell_u=\kappa_u\,\bar B_y,\; 
\ell_v=\kappa_v\,\bar B_z,\\
&\bar B_y\triangleq N_y\,d\;\mathbb{E}\!\left[\frac{1-u^2}{r}\right],\;
\bar B_z\triangleq N_z\,d\;\mathbb{E}\!\left[\frac{1-v^2}{r}\right],
\label{eq:ell_choice} 
\end{align}
where the expectations are taken under the angle and distance priors and $(\kappa_u,\kappa_v)$ are tunable constants. If such priors are not available, one can directly tune $\ell_u$ and $\ell_v$.

Let $\mathcal{D}_t=\{(i_s,\tilde z_s)\}_{s=1}^{t}$ denote the denoised observations up to slot $t$, and collect $\tilde{\mathbf{z}}_t=[\tilde z_1,\dots,\tilde z_t]^{\mathsf T}$. 
Define the kernel matrix $\mathbf{K}_t\in\mathbb{R}^{t\times t}$ with $[\mathbf{K}_t]_{a,b}=k(i_a,i_b)$, and
\begin{equation}
\widetilde{\mathbf{K}}_t \triangleq \mathbf{K}_t + \sigma_\varepsilon^2\mathbf{I}_t.
\label{eq:Ktilde}
\end{equation}
Here $\sigma_\varepsilon^2$ is an effective regularization parameter used by the GP posterior update:
it absorbs model mismatch and controls how strongly the posterior interpolates the observed debiased power
feedback. In the analysis, a conservative choice is $\sigma_\varepsilon^2 \gtrsim B_{T_1}^2$ on the bounded-noise
event in Appendix~\ref{app:stat_stage1}.

For any candidate beam $i\in\mathcal{I}$, define the kernel vector
\begin{equation}
\mathbf{k}_t(i)\triangleq[k(i,i_1),\dots,k(i,i_t)]^{\mathsf T}.
\end{equation}
Then the GP posterior on $f(i)$ is Gaussian with the following mean value and variance
\begin{align}
\mu_t(i) &= \mathbf{k}_t(i)^{\mathsf T}\widetilde{\mathbf{K}}_t^{-1}\tilde{\mathbf{z}}_t,
\label{eq:gp_post_mean}\\
\sigma_t^2(i) &= k(i,i) - \mathbf{k}_t(i)^{\mathsf T}\widetilde{\mathbf{K}}_t^{-1}\mathbf{k}_t(i).
\label{eq:gp_post_var}
\end{align}
After each new probe $(i_t,\tilde z_t)$, we augment the dataset $\mathcal{D}_t$ and update the GP posterior by recomputing the mean and variance according to \Cref{eq:gp_post_mean}-\Cref{eq:gp_post_var}. Note that the GP prior is defined on the finite DFT index set $\mathcal{I}$, and the time index $t$ merely labels the sequential observations. Multiple probes of the same beam index $i\in\mathcal{I}$ therefore correspond to multiple noisy samples of the same latent value $f(i)$, which are naturally handled by the GP posterior update.

Based on the GP posterior in \Cref{eq:gp_post_mean}-\Cref{eq:gp_post_var}, we adopt the
Gaussian-process level-set estimation (GP-LSE) principle to guide adaptive DFT probing
\cite{gotovos2013lse}. The goal of Stage~I is to efficiently identify a candidate set of
promising DFT beams with significant received powers (i.e., support) based on a limited probing budget $T_1$.

At round $t{+}1$, define the GP confidence interval for each $i\in\mathcal{I}$ as
\begin{equation}
Q_{t+1}(i)\triangleq
\Big[\mu_t(i)-\sqrt{\beta_{t+1}}\,\sigma_t(i),\;
     \mu_t(i)+\sqrt{\beta_{t+1}}\,\sigma_t(i)\Big],
\label{eq:lse_Q}
\end{equation}
where $\beta_t$ is a nondecreasing exploration schedule.
For our theoretical guarantee under the non-Gaussian debiased power noise, Appendix~\ref{app:lse_overhead} invokes the bounded-noise RKHS confidence schedule from Theorem~6 of \cite{srinivas2010gpucb}:
\begin{equation}
\beta_t
~=~
2B_f^2 + 300\,\gamma_{t-1}\log^3\!\Big(\tfrac{t}{\delta}\Big),
\; t\ge 1,
\label{eq:lse_beta_thm6_main}
\end{equation}
where $\|f\|_{\mathcal H_k}\le B_f$ (see \Cref{app:lse_overhead}) and $\gamma_t$ is the maximum information gain induced by $k$. 
The quantity $\gamma_t$ is the maximum information gain of the kernel $k$ over the finite index set $\mathcal I$, i.e., the largest mutual information between $t$ noisy samples and the latent function values that can be achieved by the best choice of $t$ query locations:
\begin{equation}\label{eq:app_gamma_def}
    \gamma_t \triangleq \max_{A\subseteq\mathcal I,\,|A|=t}\;
\frac{1}{2}\log\det\!\big(\mathbf I_t+\sigma_\varepsilon^{-2}\mathbf K_A\big),
\end{equation}
where $\mathbf K_A$ is the Gram matrix on $A$ with entries $k(\cdot,\cdot)$.
Intuitively, $\gamma_t$ quantifies the \emph{intrinsic statistical complexity} of learning $f$ under the similarity notion induced by $k$:
faster spectral decay yields smaller $\gamma_t$, hence fewer probes are needed to shrink GP confidence intervals. Some examples can be found in \Cref{tab:gamma_examples} \cite{srinivas2010gpucb}. 
Since each 1D Laplace factor $k_u$ (and $k_v$) is the Mat\'ern kernel with smoothness $\nu=\tfrac12$,
the product term in \Cref{eq:plus_laplace_kernel} is a separable 2D Mat\'ern (also $\nu=\tfrac12$) kernel,
and the overall plus-Laplace kernel is its convex combination with an additive term.
Thus, combining \cite[Theorem~4]{srinivas2010gpucb} and the Mat\'ern spectral-decay bound in
\cite[Theorem~5]{srinivas2010gpucb}, the worst-case $\gamma_t$ for our kernel follows by specializing the Mat\'ern
bound at $\nu=\tfrac12$.
Here $d$ denotes the dimension of the kernel input space and, in our Stage~I GP over the 2D DFT grid, $d=2$.
\begin{table}[t]
\centering
\caption{Typical growth of $\gamma_t$ for common kernels on a compact domain.}
\label{tab:gamma_examples}
\setlength{\tabcolsep}{15pt}
\renewcommand{\arraystretch}{1.1}
\begin{tabular}{l c}
\hline
Kernel family & Order of $\gamma_t$  \\
\hline
Linear  & $\mathcal O\!\big(d\log t\big)$ \\
Squared Exponential  & $\mathcal O\!\big((\log t)^{d+1}\big)$ \\
Mat\'ern & $\mathcal O\!\Big(t^{\frac{d(d+1)}{2\nu+d(d+1)}}\log t\Big)$ \\
\hline
\end{tabular}
\vspace{-0.5em}
\end{table}
While \Cref{eq:lse_beta_thm6_main} involves quantities that are not
available in closed form in practice, it is standard in kernelized bandit implementations to treat these terms as
tunable conservatism parameters that control the exploration-exploitation tradeoff.

Following \cite{gotovos2013lse}, we maintain a monotonically shrinking confidence region by
intersecting intervals over time:
\begin{align}
   \ell_{t+1}(i)\triangleq \max\{\ell_t(i),\ \mu_t(i)-\sqrt{\beta_{t+1}}\,\sigma_t(i)\},\nonumber\\
u_{t+1}(i)\triangleq \min\{u_t(i),\ \mu_t(i)+\sqrt{\beta_{t+1}}\,\sigma_t(i)\},
\label{eq:lse_shrink} 
\end{align}
initialized with $\ell_0(i)=-\infty$ and $u_0(i)=+\infty$. These bounds summarize all information
collected up to time $t$ and are tighter than the raw GP interval \Cref{eq:lse_Q}.

Specifically, fix a power threshold $\tau>0$ and a tolerance $\epsilon>0$. We then perform level-set classification with tolerance $\epsilon$:
\begin{align}
\mathcal{H}_{t+1}
&\triangleq
\mathcal{H}_t \cup \big\{i\in\mathcal{I}:\ \ell_{t+1}(i)>\tau-\epsilon\big\},\label{eq:lse_H}\\
\mathcal{L}_{t+1}
&\triangleq
\mathcal{L}_t \cup \big\{i\in\mathcal{I}:\ u_{t+1}(i)\le\tau+\epsilon\big\},\label{eq:lse_L}\\
\mathcal{U}_{t+1}
&\triangleq
\mathcal{I}\setminus(\mathcal{H}_{t+1}\cup\mathcal{L}_{t+1}),\label{eq:lse_U}
\end{align}
where $\mathcal{H}_t$ collects indices predicted to be above the threshold (candidate strong beams),
$\mathcal{L}_t$ collects indices predicted to be below, and $\mathcal{U}_t$ are the remaining undecided
indices that concentrate around the boundary.

To decide the next probing beam, GP-LSE selects the most ambiguous undecided index
\cite{gotovos2013lse}. Define the ambiguity score
\begin{equation}
a_{t+1}(i)\triangleq \min\big\{u_{t+1}(i)-\tau,\ \tau-\ell_{t+1}(i)\big\},\; i\in\mathcal{U}_{t+1},
\label{eq:lse_ambiguity}
\end{equation}
and choose
\begin{equation}
i_{t+1}\in \arg\max_{i\in\mathcal{U}_{t+1}} a_{t+1}(i).
\label{eq:lse_select}
\end{equation}
After probing $i_{t+1}$, we obtain $\tilde z_{t+1}$ via
\Cref{eq:stage1_power_denoised}, update the GP posterior \Cref{eq:gp_post_mean}-\Cref{eq:gp_post_var},
and repeat until we reach the Stage~I budget $T_1$ or when
$\max_{i\in\mathcal{U}_{t}} a_{t}(i)\le \epsilon$.

\subsubsection{Stage~I outputs support for Stage~II.} Suppose we have $T_1$ adaptive probes in total, the GP-LSE procedure then produces a confidence-aware partition of the DFT
index set into the predicted superlevel set $\mathcal{H}_{T_1}$, the predicted sublevel set
$\mathcal{L}_{T_1}$, and the remaining undecided set $\mathcal{U}_{T_1}$ in
\Cref{eq:lse_H}-\Cref{eq:lse_U}. In the ideal LSE setting, the natural output is the entire
superlevel set
\begin{equation}
\widehat{\mathcal{S}}_{\mathrm{LSE}}\triangleq \mathcal{H}_{T_1},
\label{eq:stage1_output_H}
\end{equation}
which has an \emph{adaptive} cardinality: as the posterior confidence region shrinks around the
threshold $\tau$, indices whose power is confidently above $\tau$ are accumulated in
$\mathcal{H}_{T_1}$, while those confidently below $\tau$ are discarded into $\mathcal{L}_{T_1}$
\cite{gotovos2013lse}. Therefore, $|\widehat{\mathcal{S}}_{\mathrm{LSE}}|$ is automatically adjusted
to the observed beamspace power profile, which is desirable for multi-path channels where the effective support size can vary across physical environments and frequency bands.

\subsubsection{Stage~I overhead in terms of information gain.}
Under the bounded-noise event $\mathcal{E}_{\mathrm{bd}}(T_1)$ in Appendix~\ref{app:stat_stage1}
and the RKHS condition $\|f\|_{\mathcal H_k}\le B_f$, the GP-LSE ambiguity
$\max_{i\in\mathcal{U}_t} a_t(i)$ drops below $\epsilon$ after a number of probes characterized by the
information gain of kernel $k$ \cite{gotovos2013lse,srinivas2010gpucb}.
In particular, using the Theorem~6 exploration schedule \Cref{eq:lse_beta_thm6_main},
a sufficient condition for termination is
\begin{equation}
T_1
= \tilde{\mathcal{O}}\!\Big(\tfrac{\beta_{T_1}\gamma_{T_1}}{\epsilon^2}\Big)
= \tilde{\mathcal{O}}\!\Big(\tfrac{B_f^2\gamma_{T_1}+\gamma_{T_1}^2}{\epsilon^2}\Big).
\label{eq:lse_overhead_scaling}
\end{equation}
Equivalently, when the algorithm is run until the
stopping condition $\max_{i\in\mathcal{U}_t} a_t(i)\le \epsilon$ (or $\mathcal{U}_t=\emptyset$) is met,
the resulting level-set estimate satisfies the standard $\epsilon$-accurate inclusion
\begin{equation}
\big\{i:\ f(i)\ge \tau+\epsilon\big\}\subseteq \mathcal{H}_{t}\subseteq \big\{i:\ f(i)\ge \tau-\epsilon\big\}
\label{eq:lse_inclusion_maintext}
\end{equation}
with probability at least $1-(\delta_{\mathrm{bd}}+\delta)$. A self-contained derivation of \Cref{eq:lse_inclusion_maintext} and \Cref{eq:lse_overhead_scaling}  is provided in Appendix~\ref{app:lse_overhead}.

\subsubsection{Practical support truncation under a fixed probing budget}
In principle, GP-LSE can be run until all indices are classified, i.e., $\mathcal{U}_t=\emptyset$, or
until the ambiguity vanishes, $\max_{i\in\mathcal{U}_t}a_t(i)\le\epsilon$ \cite{gotovos2013lse}. However,
in practical beam training we operate under a fixed probing budget $T_1$, and it is possible that a
small residual undecided set remains, i.e., $\mathcal{U}_{T_1}\neq\emptyset$, because the algorithm has
not yet fully separated all beams into $\mathcal{H}_{T_1}$ or $\mathcal{L}_{T_1}$. In this case, we
apply a lightweight \emph{fallback patch} to produce a fixed-cardinality support for Stage~II: we retain
the confidently strong beams in $\mathcal{H}_{T_1}$, and then augment/truncate using the posterior mean
ranking. Concretely, if $\mathcal{U}_{T_1}=\emptyset$, we directly set
$\widehat{\mathcal{S}}=\mathcal{H}_{T_1}$. Otherwise, when $\mathcal{U}_{T_1}\neq\emptyset$, we form
\begin{equation}
\widehat{\mathcal{S}}
~\triangleq~
\operatorname{Top}\text{-}K\Big(\{\mu_{T_1}(i): i\in\mathcal{H}_{T_1}\cup\mathcal{U}_{T_1}\}\Big),
\label{eq:stage1_topK_mu_patch}
\end{equation}
where $\operatorname{Top}\text{-}K(\cdot)$ returns the indices of the $K$ largest values. This patch is
\emph{not} used when GP-LSE has already completed the classification; it is only invoked to handle the
fixed-budget setting and to guarantee a well-defined, computationally manageable support size for the
Stage~II Gaussian-masked sensing and sparse PR refinement. The parameter $K$ can be selected using the
sparsity characterization in \Cref{sec:upa_beam_pattern} or treated as a design parameter trading
refinement complexity and robustness.
At this point, we stop exploring and convert the learned posterior into a support estimate.

\subsection{Stage II: Gaussian-Masked DFT Sensing and Rician-SPARTA Refinement}
\label{subsec:stage2}

{ Following the completion of Stage I, the BS has identified an estimate of the dominant DFT-beam support $\widehat{\mathcal{S}} \subset \mathcal{I}$, where $|\widehat{\mathcal{S}}| = K \ll N$. }
{ We define the transformation matrix}
\begin{equation}
\mathbf{F}
\triangleq
\big[
\mathbf{f}_{0,0},\,
\ldots,\,
\mathbf{f}_{N_y-1,N_z-1}
\big]^{H}
\in\mathbb{C}^{N\times N}.
\end{equation}
{ as the collection of all unit-norm 2D DFT beams as rows, ordered according to $\mathcal{I}$. Under this representation, the DFT beamspace coefficients are given by $\mathbf{s} = \mathbf{F}\mathbf{h}$, with the inverse relation $\mathbf{h} = \mathbf{F}^{H}\mathbf{s}$. }
{ Let $\mathbf{F}_{\widehat{\mathcal{S}}} \in \mathbb{C}^{K \times N}$ denote the submatrix consisting of rows of $\mathbf{F}$ indexed by $\widehat{\mathcal{S}}$, and let $\mathbf{s}_{\widehat{\mathcal{S}}} \in \mathbb{C}^{K}$ be the corresponding beamspace coefficients. As characterized in Section \ref{sec:upa_beam_pattern}, the channel can be approximated within this reduced subspace as:}
\begin{equation}
\mathbf{h}
~\approx~
\mathbf{F}_{\widehat{\mathcal{S}}}^{H}\,\mathbf{s}_{\widehat{\mathcal{S}}},
\;
\|\mathbf{s}_{\widehat{\mathcal{S}}}\|_{0}\leq K,
\label{eq:stage2_beamspace_model}
\end{equation}
with sparsity level \(K\). 


{ In Stage II, the BS constructs a second set of $M_2$ pilot beams by linearly combining only the DFT beams within the estimated support $\widehat{\mathcal{S}}$. For each measurement index $p=1, \dots, M_2$, we draw an i.i.d. complex Gaussian mask $\mathbf{g}_p \sim \mathcal{CN}(\mathbf{0}, \frac{1}{K}\mathbf{I}_K)$ and synthesize the Gaussian-masked DFT beam:}
\begin{equation}
\mathbf{v}_p
~=~
\mathbf{F}_{\widehat{\mathcal{S}}}^{H}\mathbf{g}_p
\in\mathbb{C}^{N}.
\label{eq:gauss_masked_beam}
\end{equation}
The BS transmits a unit-modulus pilot through \(\mathbf{v}_p\), and the
user feeds back only the received amplitude (or power). The resulting
amplitude-only observation at slot \(p\) is
\begin{equation}
y_p ~=~ \big|\mathbf{h}^{H}\mathbf{v}_p + w_p\big|,
\;
w_p\sim\mathcal{CN}(0,\sigma^2),
\label{eq:stage2_rician_meas}
\end{equation}
{ where $y_p$ follows a Rician distribution with non-centrality parameter $|\mathbf{h}^{H}\mathbf{v}_p|$. By substituting \Cref{eq:stage2_beamspace_model} and \Cref{eq:gauss_masked_beam} into \Cref{eq:stage2_rician_meas}, we obtain an equivalent low-dimensional beamspace model:}
\begin{equation}
y_p
~=~
\big|\mathbf{s}_{\widehat{\mathcal{S}}}^{H}\mathbf{g}_p + w_p\big|,
\;
p=1,\dots,M_2.
\label{eq:stage2_lowdim_model}
\end{equation}
Hence, Stage~II reduces to recovering the sparse beamspace vector
\(\mathbf{s}_{\widehat{\mathcal{S}}}\) from Rician-corrupted magnitude
measurements with i.i.d. complex Gaussian sensing
vectors \(\{\mathbf{g}_p\}_{p=1}^{M_2}\).

{ Directly applying standard PR solvers to $y_p$ often leads to performance degradation due to the bias introduced by the noise in the magnitude domain. To mitigate this, we observe that $\mathbb{E}[y_p^2 \,|\, \mathbf{g}_p] = |\mathbf{s}_{\widehat{\mathcal{S}}}^{H} \mathbf{g}_p|^2 + \sigma^2$. We therefore perform Rician denoising by defining the pseudo-amplitude $\tilde{\psi}_p$:}
\begin{equation}
\tilde{\psi}_p
~\triangleq~
\sqrt{\big[y_p^2-\sigma^2\big]_+}.
\label{eq:rician_denoise}
\end{equation}

The resulting
\(\tilde{\psi}_p\) serves as a {pseudo-amplitude} that approximates
\(|\mathbf{s}_{\widehat{\mathcal{S}}}^{H}\mathbf{g}_p|\) while absorbing the residual Rician fluctuations into an effective noise term. This
yields an amplitude-based sparse PR model of the form
\begin{equation}
\tilde{\psi}_p
~=~
\big|\mathbf{s}_{\widehat{\mathcal{S}}}^{H}\mathbf{g}_p\big|
+\tilde{\eta}_p,
\; p=1,\dots,M_2,
\label{eq:stage2_amp_model}
\end{equation}
i.e., in our setting, we apply the Rician debiasing and the pseudo-amplitude mapping \Cref{eq:rician_denoise}, which yields an equivalent noisy-amplitude model \Cref{eq:stage2_amp_model} (see proof in \Cref{app:stat_stage2}).
Therefore, Stage~II can be viewed as running a robust truncated amplitude-flow solver on noisy amplitudes \cite{wang2018sparta_tsp}.

Given the denoised amplitudes \(\{\tilde{\psi}_p\}\) and sensing vectors
\(\{\mathbf{g}_p\}\), we seek a \(k\)-sparse estimate of
\(\mathbf{s}_{\widehat{\mathcal{S}}}\) by minimizing an amplitude-based
least-squares loss subject to sparsity:
\begin{equation}
\min_{\mathbf{z}\in\mathbb{C}^{K}}
\frac{1}{M_2}\sum_{p=1}^{M_2}
\big(\tilde{\psi}_p-\big|\mathbf{g}_p^{H}\mathbf{z}\big|\big)^2
\;
\text{s.t. }\|\mathbf{z}\|_0\leq k\leq K.
\label{eq:stage2_sparta_objective}
\end{equation}

We adopt a SPARTA-style solver\cite{wang2018sparta_tsp} that combines truncated
amplitude-based gradient steps with hard thresholding, but with two
modifications tailored to our setting: {(i)} the algorithm
operates only on the reduced beamspace supported on
\(\widehat{\mathcal{S}}\), and {(ii)} the amplitudes in
\Cref{eq:stage2_sparta_objective} are replaced by the Rician-denoised
pseudo-measurements \(\tilde{\psi}_p\) in \Cref{eq:rician_denoise}.

Concretely, let \(\mathbf{z}^{(0)}\) be an initialization restricted to
\(\widehat{\mathcal{S}}\), which can be got by a scaled principal eigenvector of a
sample covariance matrix formed from the largest \(\tilde{\psi}_p\)'s,
and denote \(\mathcal{I}_t\subseteq\{1,\dots,M_2\}\) as the truncation
set at iteration \(t\), which retains only those indices whose amplitude residuals are not too large. A generic Rician-SPARTA refinement
iteration takes the form
\begin{align}
\mathbf{z}^{(t+\frac{1}{2})}
&= \mathbf{z}^{(t)}
- \mu_t
\frac{1}{|\mathcal{I}_t|}
\sum_{p\in\mathcal{I}_t}
\Big(\big|\mathbf{g}_p^{H}\mathbf{z}^{(t)}\big|
      -\tilde{\psi}_p\Big)
\frac{\mathbf{g}_p\mathbf{g}_p^{H}\mathbf{z}^{(t)}}
     {\big|\mathbf{g}_p^{H}\mathbf{z}^{(t)}\big|},
\label{eq:stage2_trunc_grad}
\\
\mathbf{z}^{(t+1)}
&= \mathcal{H}_k\big(\mathbf{z}^{(t+\frac{1}{2})}\big),
\label{eq:stage2_hard_threshold}
\end{align}
where \(\mu_t>0\) is a stepsize, and \(\mathcal{H}_k(\cdot)\) denotes the hard-thresholding
operator that keeps the \(k\) entries with the largest magnitudes and sets the rest to zero. The truncation rule defining \(\mathcal{I}_t\) follows
the standard SPARTA design \cite{wang2018sparta_tsp}, with thresholds applied to
the amplitude residuals
\(\big|\mathbf{g}_p^{H}\mathbf{z}^{(t)}\big|-\tilde{\psi}_p\) to
suppress the impact of outliers and residual Rician noise.

After a prescribed number of iterations, we obtain
\(\widehat{\mathbf{s}}_{\widehat{\mathcal{S}}}=\mathbf{z}^{(T)}\) and
form the final channel estimate
\begin{equation}
\widehat{\mathbf{h}}
~=~
\mathbf{F}_{\widehat{\mathcal{S}}}^{H}\,
\widehat{\mathbf{s}}_{\widehat{\mathcal{S}}},
\label{eq:stage2_final_h}
\end{equation}
which is then used for downlink beamforming. 

Compared with applying
SPARTA on the full DFT beamspace directly, the proposed Stage~II operates on a much smaller subspace identified by Stage~I and incorporates Rician
denoising via \Cref{eq:rician_denoise}, thereby reducing both the
required number of measurements and the computational cost, while better
matching the amplitude-only feedback model in near-field beam training.


{ Specializing the SPARTA theory \cite{wang2018sparta_tsp} to the ambient dimension $K$ (where $K = |\widehat{\mathcal{S}}|$ is the size of the discovered index set) implies that an additional $M_2$ measurements on the order of}
\begin{equation}
M_2 \;\gtrsim\; C\,k^2\log(K)
\label{eq:stage2_m2_scaling_text}
\end{equation}
are sufficient for exact recovery with high probability.

Importantly, Stage~II probes only $\mathrm{span}(\mathbf{F}_{\widehat{\mathcal{S}}}^{H})$, hence the noiseless
observations depend only on $\mathbf{s}_{\widehat{\mathcal{S}}}$ regardless of energy outside $\widehat{\mathcal{S}}$.
Consequently, for any estimate $\widehat{\mathbf{s}}$ supported on $\widehat{\mathcal{S}}$ and
$\widehat{\mathbf{h}}=\mathbf{F}^{H}\widehat{\mathbf{s}}$, unitarity of $\mathbf{F}$ yields the exact error decomposition
\begin{equation}
\|\widehat{\mathbf{h}}-\mathbf{h}\|_2^2
=\|\widehat{\mathbf{s}}_{\widehat{\mathcal{S}}}-\mathbf{s}_{\widehat{\mathcal{S}}}\|_2^2
+\|\mathbf{s}_{\widehat{\mathcal{S}}^c}\|_2^2,
\label{eq:error_decomp_support_text}
\end{equation}
which separates the algorithmic error on the discovered subspace and the irreducible tail caused by
false negatives. { False positives primarily enlarge $K$, thereby increasing the logarithmic term in \Cref{eq:stage2_m2_scaling_text} and raising computational cost, whereas the tail energy $\|\mathbf{s}_{\widehat{\mathcal{S}}^c}\|_2$ imposes a fundamental error floor when $K$ fails to cover the full support set. }
However, under Stage~I, there is a high probability that $\widehat{\mathcal{S}}$ coincides with the true dominant DFT support \cite{srinivas2012information}; consequently, the tail term $\|\mathbf{s}_{\widehat{\mathcal{S}}^c}\|_2$ in \Cref{eq:error_decomp_support_text} vanishes.

\begin{algorithm}[t]
\caption{Near-Field Beam Training via GP-LSE Assisted Sparse Phase Retrieval}
\label{alg:gb_spr}
\begin{algorithmic}
\small
\STATE \textbf{Input:} 2D DFT codebook $\{\mathbf{f}_{n,m}\}$ and transform $\mathbf{F}$; budgets $(T_1,M_2)$; noise variance $\sigma^2$;
kernel $k(\cdot,\cdot)$; LSE parameters $(\tau,\epsilon,\{\beta_t\})$; sparsity level $k$ for Stage~II.
\STATE \textbf{Output:} $\widehat{\mathcal{S}}$, $\widehat{\mathbf{s}}_{\widehat{\mathcal{S}}}$, $\widehat{\mathbf{h}}$, $\widehat{\mathbf{v}}$.

\vspace{0.2em}
\STATE \textbf{Stage I (GP-LSE on 2D DFT grid):} initialize $(\mathcal{H}_0,\mathcal{L}_0,\mathcal{U}_0)$ and $(\ell_0(\cdot),u_0(\cdot))$ as in \Cref{eq:lse_shrink}, with $\mathcal{D}_0=\emptyset$.
\FOR{$t=1,\dots,T_1$}
    \STATE Update GP posterior $(\mu_{t-1}(\cdot),\sigma_{t-1}(\cdot))$ via \Cref{eq:gp_post_mean}-\Cref{eq:gp_post_var}.
    \STATE Update $(\ell_t(\cdot),u_t(\cdot))$ and $(\mathcal{H}_t,\mathcal{L}_t,\mathcal{U}_t)$ via \Cref{eq:lse_shrink}-\Cref{eq:lse_U}.
    \STATE If $\mathcal{U}_t=\emptyset$ or $\max_{i\in\mathcal{U}_t}a_t(i)\le\epsilon$ with $a_t(\cdot)$ in \Cref{eq:lse_ambiguity}, \textbf{break}.
    \STATE Select $i_t\in\arg\max_{i\in\mathcal{U}_t} a_t(i)$ via \Cref{eq:lse_select}, probe $\mathbf{v}_{i_t}=\mathbf{f}_{i_t}$, and form $\tilde z_t$ via \Cref{eq:stage1_power_denoised}.
    \STATE Append $(i_t,\tilde z_t)$ to $\mathcal{D}_t$.
\ENDFOR
\STATE \textbf{Support:} set $\widehat{\mathcal{S}}=\mathcal{H}_{\widehat{T}_1}$ if $\mathcal{U}_{\widehat{T}_1}=\emptyset$, else apply the Top-$K$ patch \Cref{eq:stage1_topK_mu_patch}.

\vspace{0.2em}
\STATE \textbf{Stage II (Gaussian-masked DFT \& Rician-SPARTA):} let $K=|\widehat{\mathcal{S}}|$ and form $\mathbf{F}_{\widehat{\mathcal{S}}}$.
\FOR{$p=1,\dots,M_2$}
    \STATE Draw $\mathbf{g}_p\sim\mathcal{CN}(\mathbf{0},\frac{1}{K}\mathbf{I})$ and set $\mathbf{v}_p=\mathbf{F}_{\widehat{\mathcal{S}}}^{H}\mathbf{g}_p$ \Cref{eq:gauss_masked_beam}.
    \STATE Observe $y_p$ in \Cref{eq:stage2_rician_meas} and compute $\tilde{\psi}_p$ via \Cref{eq:rician_denoise}.
\ENDFOR
\STATE Obtain $\widehat{\mathbf{s}}_{\widehat{\mathcal{S}}}$ by running the Rician-SPARTA solver defined in \Cref{eq:stage2_sparta_objective} (updates \Cref{eq:stage2_trunc_grad}-\Cref{eq:stage2_hard_threshold}).
\STATE Recover $\widehat{\mathbf{h}}$ by \Cref{eq:stage2_final_h} and set $\widehat{\mathbf{v}}=\widehat{\mathbf{h}}/\|\widehat{\mathbf{h}}\|_2$.
\STATE \textbf{return} $\widehat{\mathcal{S}},\,\widehat{\mathbf{s}}_{\widehat{\mathcal{S}}},\,\widehat{\mathbf{h}},\,\widehat{\mathbf{v}}$.
\end{algorithmic}
\end{algorithm}
\section{Simulation Results}\label{sec:simulation}
{ This section evaluates the performance of the proposed two-stage near-field beam training framework through comprehensive numerical simulations. We consider a single BS equipped with a UPA of $N_y=128$ and $N_z=16$ antenna elements, operating at a carrier frequency of $f_c=28$~GHz. The total number of antenna elements is $N=N_yN_z=2048$. For each method, we reconstruct the channel estimate $\hat{\mathbf{h}}$ and evaluate the normalized correlation:}
\begin{equation}
\rho \triangleq \frac{\left|\mathbf{h}^{\mathsf{H}}\hat{\mathbf{h}}\right|}{\|\mathbf{h}\|_2\,\|\hat{\mathbf{h}}\|_2}\in[0,1],
\end{equation}
where $\rho \approx 1$ signifies near-perfect channel recovery. All results are averaged over $500$ Monte-Carlo trials.

\subsection{{ Performance under Varying SNR}
}\label{subsec:snr_sweep}
{ We first assess robustness against additive noise by sweeping the SNR from $-15$~dB to $15$~dB. Specifically, we set the noise power as
\(
\frac{\|h\|_2^2}{n\,\mathrm{SNR}},\text{ and }
\mathrm{SNR}=10^{\mathrm{SNR}_{\mathrm{dB}}/10},
\)
where $n$ is the channel dimension and $h$ is the user channel. In each trial, $L=6$ paths are generated with
normalized angular variables $u,v\sim\mathcal{U}(-0.5,0.5)$ and the propagation range $r$ drawn uniformly from the
near-field region $[r_{\mathrm{F}},\,r_{\mathrm{R}}]$.
These parameters are randomized across $500$ trials to calculate the mean channel correlation.
{ We compare the proposed LSE+R-SPARTA against three baselines:
\emph{(i) Exhaustive DFT Search}: A classical baseline that exhaustively probes all $N$ DFT beams and selects the one with the highest received power (pilot overhead: $N$).
\emph{(ii) R-SPARTA} and \emph{(iii) R-SWF}: Sparse phase retrieval solvers applied directly to the full DFT beamspace using Gaussian-masked sensing vectors.
\emph{(iv) LSE+R-SPARTA (Proposed)}: Our two-stage framework. Stage I performs structured exploration to identify a dominant beam subset ($T_1 = N$), and Stage II executes sparse PR on the resulting subspace ($M_2 = N$). Except for the exhaustive search, all PR-based methods utilize a total training overhead of $2N$. }

{ The results are illustrated in \Cref{fig:snr_sweep}. The proposed LSE+R-SPARTA achieves consistently high correlation across the entire SNR range, with performance improving monotonically from $0.693$ at $-15$~dB to $0.957$ at $15$~dB. The performance gain is particularly prominent in the low-SNR regime, where direct sparse phase retrieval is highly susceptible to noise-induced ambiguities and suboptimal initializations. By performing adaptive subset selection in Stage I, we effectively improve the conditioning and local SNR of the subsequent recovery, leading to more reliable convergence. For instance, at $-15$~dB, our method yields a mean correlation of $0.693$ compared to the $0.425$ achieved by exhaustive DFT, reaching an absolute gain of $0.268$. At $0$~dB, our method reaches $0.877$, whereas the non-adaptive R-SPARTA baseline only attains $0.648$. At high SNR, our framework remains highly competitive, slightly outperforming R-SWF at $15$~dB ($0.957$ vs. $0.952$). These results validate that the proposed two-stage design significantly bolsters robustness without sacrificing high-SNR accuracy.}
\begin{figure}[!t]
\centering
\includegraphics[width=0.7\linewidth]{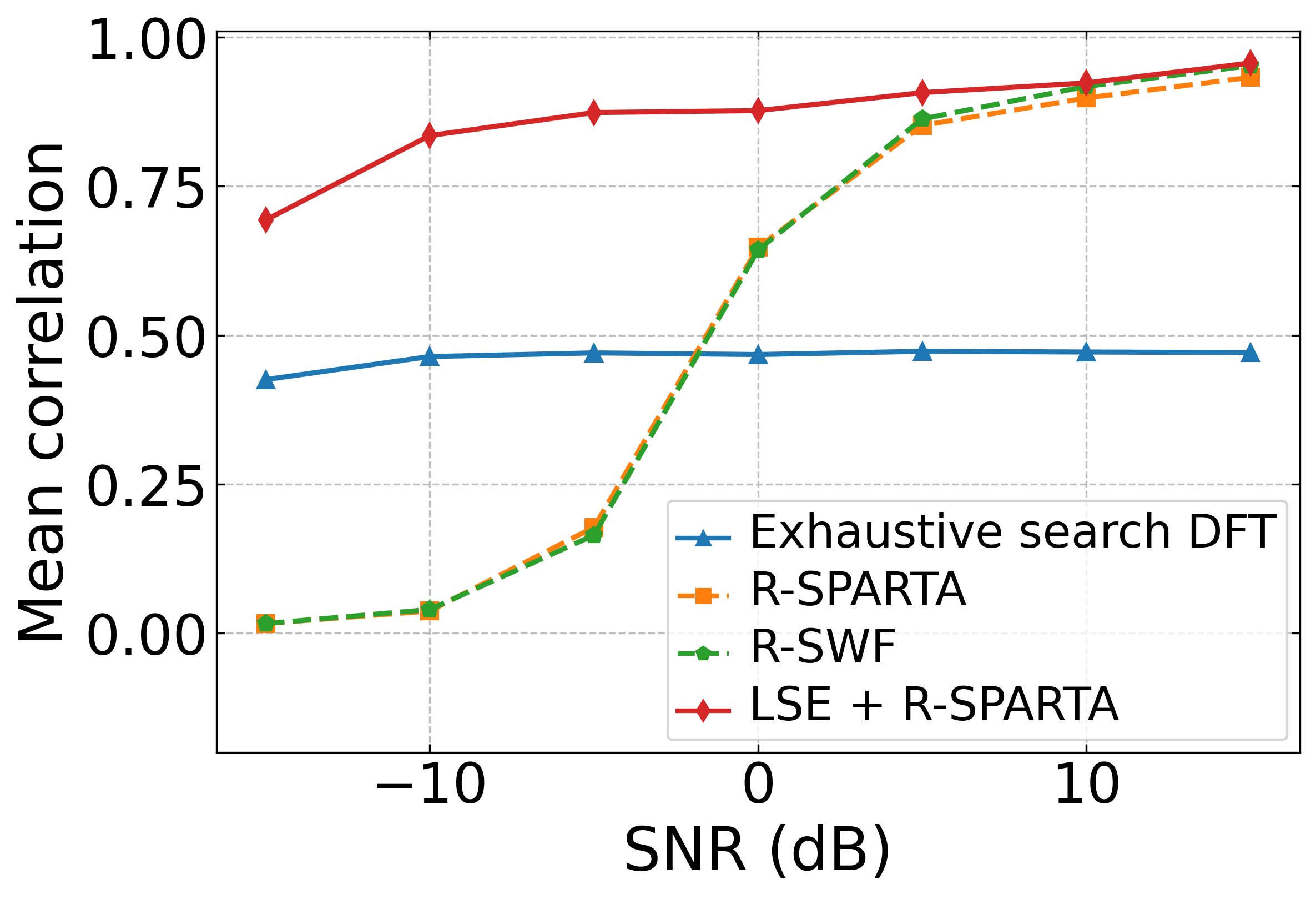}
\caption{Mean correlation versus SNR.}
\label{fig:snr_sweep}
\end{figure}

\subsection{Ablation Study}
{ To quantify the individual contributions of the key components within our two-stage framework, we perform an ablation study using the SNR sweep configuration described in \Cref{subsec:snr_sweep}. We evaluate the performance of the proposed plus-kernel LSE+R-SPARTA against two modified variants, reporting the mean channel correlation at each SNR point as illustrated in \Cref{fig:ablation}.

\textbf {Impact of Rician-Aware Recovery}: We first isolate the significance of the statistical refinement by comparing our framework against a plus-kernel LSE+SPARTA variant. This version disables Stage~I denoising and replaces the Rician-aware Stage~II solver with a standard SPARTA algorithm. As shown in the results, removing the Rician awareness leads to a consistent degradation in correlation across the entire SNR spectrum. Specifically, at $-15$~dB, the mean correlation falls from $0.693$ to $0.575$, while at $15$~dB, it decreases from $0.957$ to $0.945$. These findings confirm that integrating denoising and Rician-tailored recovery is critical for maintaining robust gains in practical non-coherent scenarios.

\textbf {Impact of Physically-Motivated Kernel Design}: We next examine the efficacy of the Stage~I kernel by switching to a Laplace-kernel LSE+R-SPARTA variant. This setup retains the Rician-aware recovery pipeline but replaces the proposed cross-shaped kernel with a generic Laplace kernel. While this modification maintains competitive performance at high SNR ($0.956$ vs. $0.957$ at $15$~dB), it suffers significantly in low-SNR regimes, where the mean correlation drops from $0.693$ to $0.565$ at $-15$~dB. The superior robustness of the plus-shaped kernel at low SNR stems from its ability to accurately model the cross-shaped spatial reward map of the near-field DFT beamspace. By providing a more accurate inductive bias during the exploration phase, the plus kernel enables the BS to identify a more reliable beam support under heavy noise, thereby providing a better-conditioned subspace for the subsequent refinement stage.

}

\begin{figure}[!t]
\centering
\includegraphics[width=0.7\linewidth]{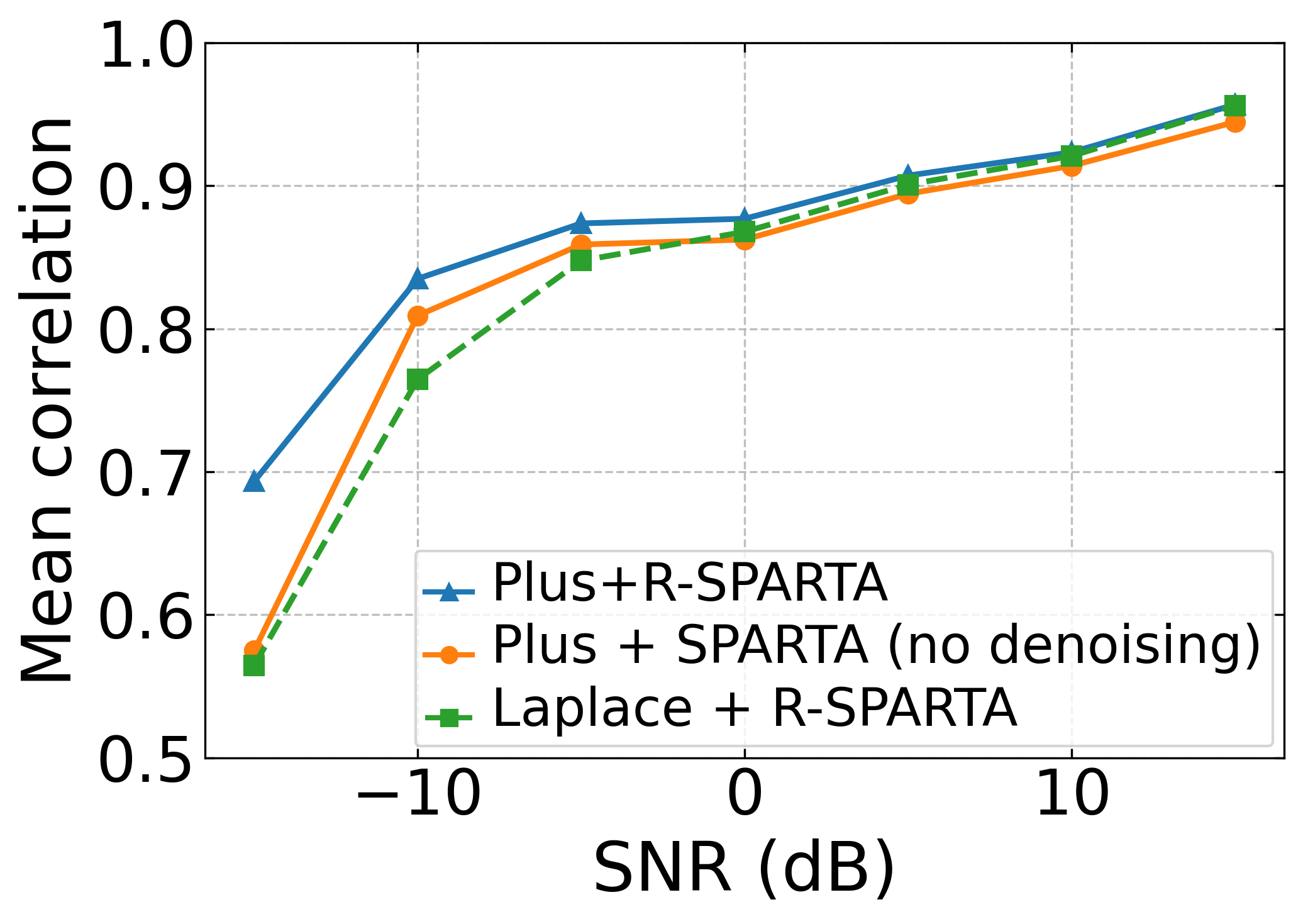}
\caption{Mean correlation versus SNR for ablated methods.}
\label{fig:ablation}
\end{figure}

\tightsubsection{Impact of Number of Paths}\label{subsec:path}
We next evaluate the resilience of the proposed framework to varying channel sparsity by sweeping the number of propagation paths $L \in \{2, 4, 6, 8, 10, 12\}$. All other simulation parameters remain consistent with the setup in \Cref{subsec:snr_sweep}. The resulting mean channel correlation $\rho$ is illustrated in \Cref{fig:path}. As the multipath environment becomes increasingly complex (higher $L$), the baseline methods exhibit significant performance degradation. For the exhaustive DFT search, the mean correlation drops monotonically from 0.607 at $L=2$ to 0.396 at $L=12$. This sharp decline confirms that selecting a single best DFT beam is fundamentally inadequate for capturing the diverse spatial energy of a rich multipath environment. Similarly, the performance of R-SPARTA and R-SWF declines to 0.700 and 0.791 at $L=12$, respectively. In contrast, LSE+R-SPARTA demonstrates remarkable stability, maintaining a correlation above 0.915 across all tested values of $L$. Notably, the performance even improves slightly as $L$ increases, reaching 0.942 at $L=12$. This translates to a relative improvement of approximately 34.5 $\%$ over R-SPARTA and 19.0 $\%$  over R-SWF in the most complex scenario ($L=12$). 

These results underscore the strength of our two-stage design: the structured exploration in Stage I adaptively identifies a beam subset that remains highly informative even as the multipath richness grows. By effectively "masking" the relevant subspace, the framework provides a stabilized environment for the subsequent Stage II recovery, whereas non-adaptive methods struggle to resolve the increased interference between paths.



\begin{figure}[!t]
\centering
\includegraphics[width=0.67\linewidth]{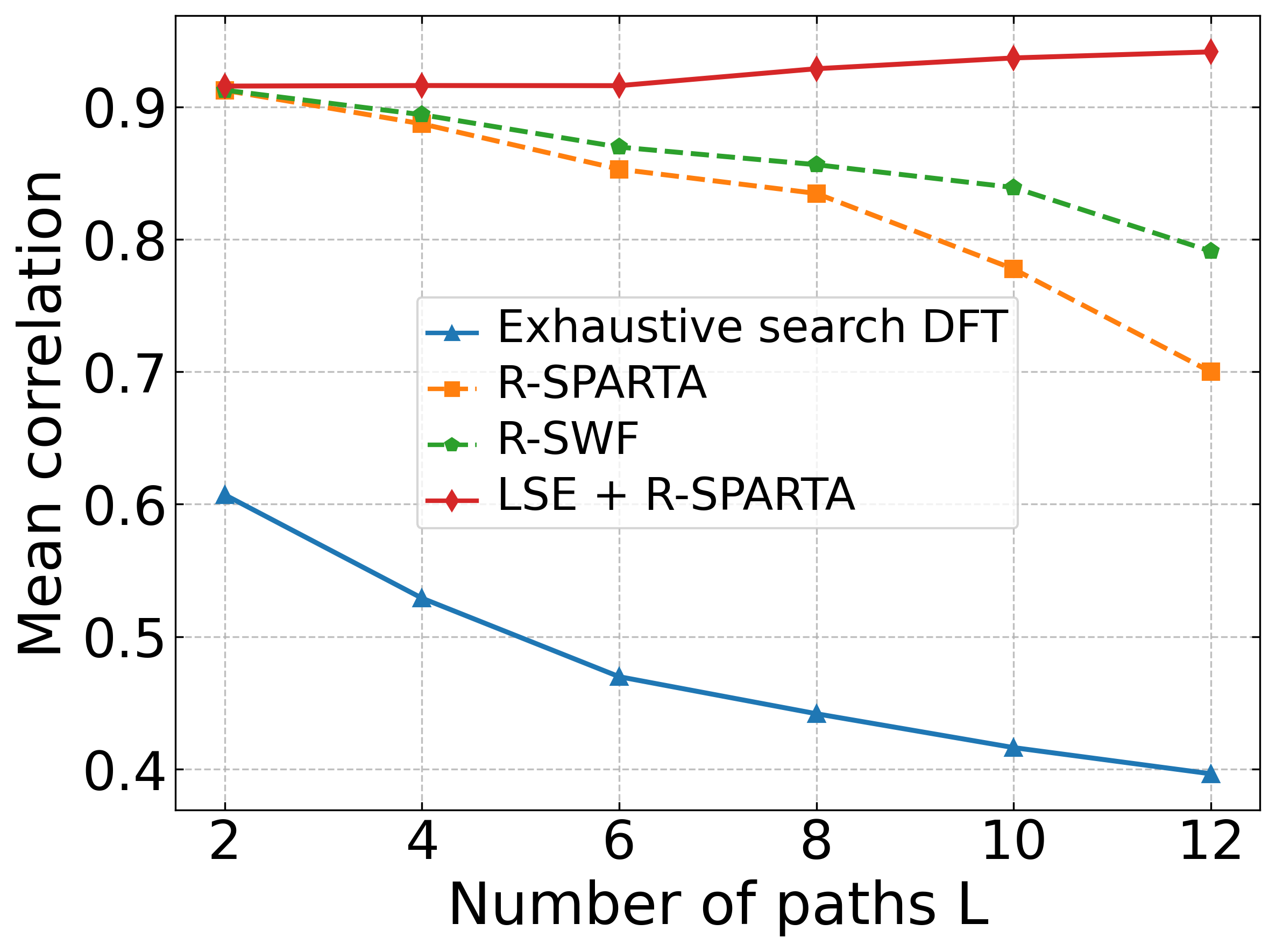}
\caption{Mean correlation versus number of path.}
\label{fig:path}
\end{figure}
\subsection{Impact of User Distance and Near-Field Robustness}\label{subsec:distance}
{ We next evaluate the resilience of the proposed framework to varying user-BS distances by sweeping the user range from $10$ m to $80$ m in increments of $10$ m. For this experiment, the multipath complexity is fixed at $L=4$ paths. At each distance point, we independently randomize the angles and scattering parameters for all paths, ensuring the user's primary range matches the prescribed grid value. The resulting mean channel correlation $\rho$ is illustrated in \Cref{fig:distance}.Across the entire evaluation range, LSE+R-SPARTA maintains a consistently high performance that is largely insensitive to distance variations, with the mean correlation stabilized near $0.95$. In contrast, the R-SPARTA and R-SWF baselines achieve lower correlation values (approximately $0.91$) and exhibit a slight performance degradation as the distance increases. The exhaustive DFT search remains the least effective baseline, showing significantly lower correlation and higher sensitivity to distance changes.

Specifically, at a close range of $10$ m, LSE+R-SPARTA achieves a correlation of $0.955$, outperforming R-SPARTA ($0.917$) by $4.2\%$ and the exhaustive DFT search ($0.459$) by $108.3\%$. This superior performance persists even at a larger range of $80$ m, where our method maintains a correlation of $0.955$, exceeding R-SPARTA ($0.907$) by $5.3\%$ and the exhaustive search ($0.527$) by $81.2\%$. These findings demonstrate that the structured exploration in Stage I effectively identifies a robust beam subset that captures the spherical wavefront characteristics reliably across a broad span of near-field distances.

}


\begin{figure}[!t]
\centering
\includegraphics[width=0.67\linewidth]{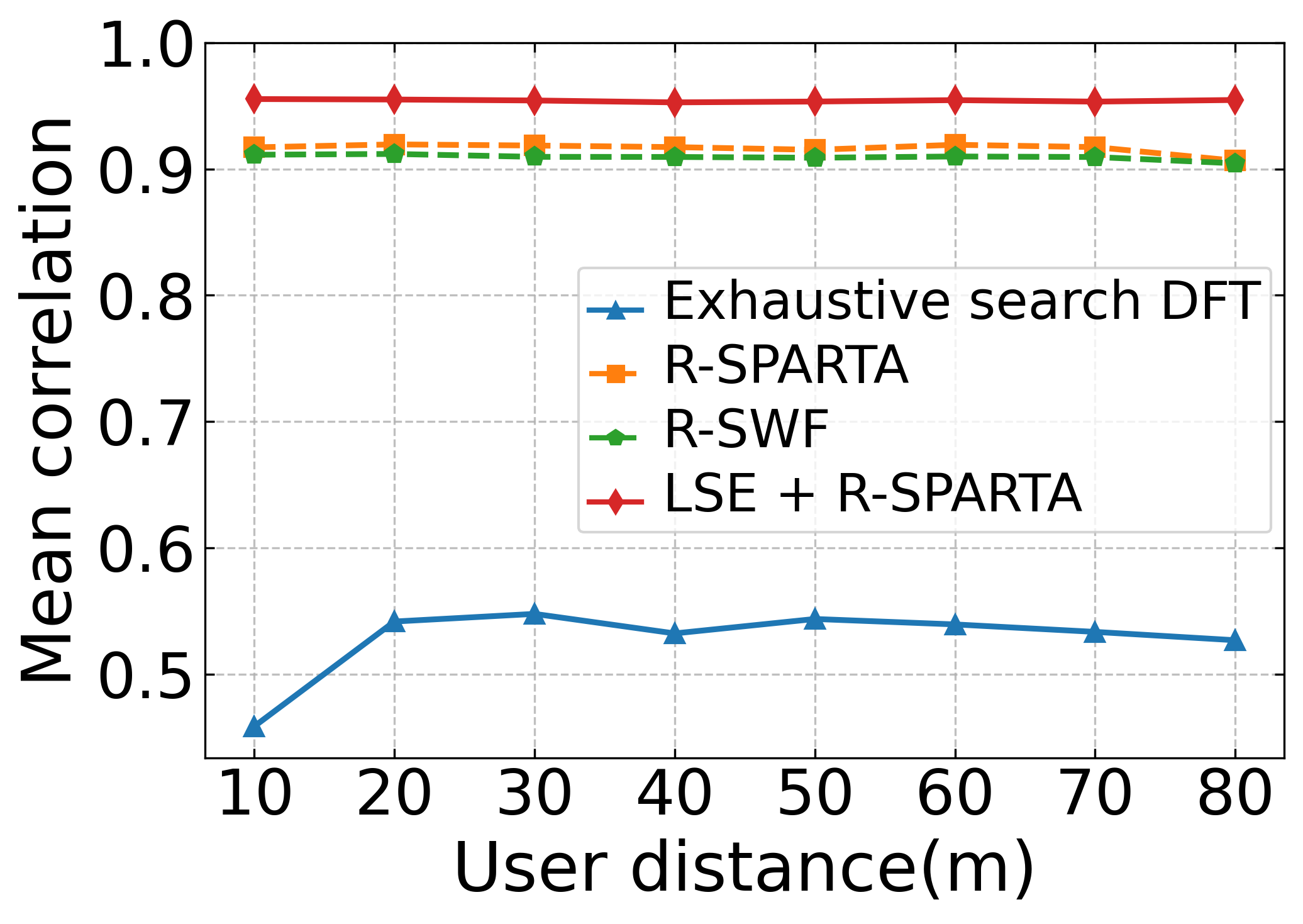}
\caption{Mean correlation versus user distance.}
\label{fig:distance}
\end{figure}

\section{Conclusion}\label{sec:clu}
This paper proposed a two-stage framework for amplitude-only near-field beam training in multipath channels. 
Stage~I performs GP-LSE-guided structured exploration on the standard 2D DFT grid using a physics-motivated cross-pattern-kernel that matches the near-field beamspace locality, enabling efficient discovery of dominant DFT support from debiased power feedback. 
Stage~II then refines the channel over the learned subspace via Gaussian-masked sparse phase retrieval with a Rician-aware pseudo-amplitude denoising step, improving robustness under non-coherent measurements. 
Numerical results show that the proposed method consistently outperforms non-adaptive baselines, achieving a $35.3\%$ relative improvement in channel correlation at moderate SNR, and remains stable across different multipath complexities and user distances. 
Overall, the proposed design offers a scalable and practical solution for near-field beam training when phase coherence cannot be guaranteed.

\appendices
\crefalias{section}{appendix}
\section{Derivation of Sparsity Level for UPA}\label{app:sparsity_upa_derivation}
Under the separable lobe approximation, the number of active 2D DFT indices per path equals the occupied area divided by the grid cell size:
\begin{equation}
K_{\text{per-path}}
~\approx~
\frac{B_y(u,r)\,B_z(v,r)}{\Delta_y\Delta_z}
=
N_yN_zd^2\frac{(1-u^2)(1-v^2)}{r^2\,\Delta_y\Delta_z}.
\end{equation}
For $L$ disjoint paths, $\mathbb{E}[K_{\mathrm{UPA}}]\approx L\,\mathbb{E}[K_{\text{per-path}}]$, which gives \Cref{eq:sparsity_upa_general}.

Because $u=\sqrt{1-v^2}\,s$ with $s=\sin\phi$, we compute
\begin{align}
\mathbb{E}\!\left[(1-u^2)(1-v^2)\right]
&=\mathbb{E}\!\left[(1-v^2)\left(1-(1-v^2)s^2\right)\right]\nonumber\\
&=\mathbb{E}[1-v^2]-\mathbb{E}\!\left[(1-v^2)^2\right]\mathbb{E}[s^2],
\end{align}
where we used the independence $v\perp s$. With $v\sim\mathcal{U}[v_1,v_2]$ and $s\sim\mathcal{U}[s_1,s_2]$,
\begin{align}
\mathbb{E}[v^2]&=\frac{v_2^3-v_1^3}{3(v_2-v_1)}\triangleq \mu_2,\;
\mathbb{E}[v^4]=\frac{v_2^5-v_1^5}{5(v_2-v_1)}\triangleq \mu_4,\nonumber\\
\mathbb{E}[s^2]&=\frac{s_2^3-s_1^3}{3(s_2-s_1)}\triangleq \nu_2,
\end{align}
and
\[
\mathbb{E}[1-v^2]=1-\mu_2,\;
\mathbb{E}[(1-v^2)^2]=1-2\mu_2+\mu_4,
\]
which yields \Cref{eq:Xi_us}. Finally, since $r\sim\mathcal{U}[r_1,r_2]$,
\[
\mathbb{E}\!\left[\frac{1}{r^2}\right]
=
\frac{1}{r_2-r_1}\int_{r_1}^{r_2}\frac{1}{r^2}\,dr
=\frac{1}{r_2-r_1}\left(\frac{1}{r_1}-\frac{1}{r_2}\right),
\]
which gives \Cref{eq:Xi_r}. Substituting these into \Cref{eq:sparsity_upa_general} yields \Cref{eq:sparsity_upa_closed}.

\section{Unbiasedness and noise statistics of Stage~I power feedback}
\label{app:stat_stage1}

Let $y_t=s_t+w_t$ with $s_t\triangleq \mathbf{h}^H\mathbf{v}_{i_t}x_t$ and
$w_t\sim\mathcal{CN}(0,\sigma^2)$. Expanding the squared magnitude yields
\begin{equation}
|y_t|^2 = |s_t|^2 + |w_t|^2 + 2\Re\{s_t^{*}w_t\}.
\end{equation}
Therefore, conditioning on $\mathbf{h}$ (equivalently on $s_t$) gives
\begin{equation}
\mathbb{E}\!\left[|y_t|^2\,\big|\,\mathbf{h}\right]
= |s_t|^2 + \mathbb{E}[|w_t|^2]
= f(i_t)+\sigma^2,
\end{equation}
which justifies $\tilde z_t=|y_t|^2-\sigma^2$ as a conditionally unbiased estimate of $f(i_t)$.
The centered residual admits the decomposition
\begin{equation}
\tilde\varepsilon_t
= \tilde z_t - f(i_t)
= (|w_t|^2-\sigma^2) + 2\Re\{s_t^{*}w_t\}.
\label{eq:app_stage1_residual_decomp}
\end{equation}

For $w_t\sim\mathcal{CN}(0,\sigma^2)$, $|w_t|^2$ follows an exponential distribution with mean
$\sigma^2$ and variance $\sigma^4$, hence $(|w_t|^2-\sigma^2)$ is sub-exponential (light-tailed).
Moreover, $2\Re\{s_t^{*}w_t\}$ is real Gaussian with zero mean and variance $2\sigma^2|s_t|^2
=2\sigma^2 f(i_t)$. Consequently, the conditional variance of the power feedback is
\begin{equation}
\mathrm{Var}\!\left(|y_t|^2\,\big|\,\mathbf{h}\right)
=\sigma^4 + 2\sigma^2 f(i_t),
\label{eq:app_stage1_power_var}
\end{equation}
which is heteroscedastic across probed beams.

We further assume that the beamspace power map is uniformly bounded,
\begin{equation}\label{eq:app_fmax_assump}
0\le f(i)\le f_{\max},\; \forall i\in\mathcal{I},
\end{equation}
where $f_{\max}$ is a finite constant induced by practical transmit-power and array-gain limits.
Under \Cref{eq:app_fmax_assump}, \Cref{eq:app_stage1_power_var} implies the uniform variance bound
\begin{equation}\label{eq:app_stage1_var_ub}
\mathrm{Var}\!\left(\tilde\varepsilon_t\,\big|\,\mathbf{h}\right)
=\sigma^4+2\sigma^2 f(i_t)
\le \sigma^4+2\sigma^2 f_{\max}.
\end{equation}

The residual $\tilde\varepsilon_t$ in \Cref{eq:app_stage1_residual_decomp} is a sum of a centered exponential term
and a Gaussian term, hence it is conditionally sub-exponential. In particular, there exist absolute constants
$c_0,c_1>0$ such that for all $u\ge 0$,
\begin{equation}\label{eq:app_stage1_subexp_tail}
\Pr\!\Big(|\tilde\varepsilon_t|\ge u \,\big|\,\mathbf{h}\Big)
\le
2\exp\!\Bigg(
-c_0 \min\Big\{
\tfrac{u^2}{\sigma^4+2\sigma^2 f_{\max}},
\ \tfrac{u}{\sigma^2}
\Big\}
\Bigg).
\end{equation}

We next use the sub-exponential tail bound \Cref{eq:app_stage1_subexp_tail} to construct a high-probability
bounded-noise event over a finite Stage~I budget $T_1$. This event is only an analysis device; the
implementation always uses the original debiased feedback $\tilde z_t$.

Fix $\delta_{\mathrm{bd}}\in(0,1)$ and define, for $t=1,\dots,T_1$, the deterministic threshold
\begin{equation}\label{eq:app_Bt_def}
B_t \triangleq c_1\Bigg(
\sigma^2\log\!\Big(\tfrac{\pi^2 t^2}{12\delta_{\mathrm{bd}}}\Big)
+
\sqrt{\big(\sigma^4+2\sigma^2 f_{\max}\big)\log\!\Big(\tfrac{\pi^2 t^2}{12\delta_{\mathrm{bd}}}\Big)}
\Bigg),
\end{equation}
where $c_1>0$ is an absolute constant as in \Cref{eq:app_stage1_subexp_tail}. Define the event
\begin{equation}\label{eq:app_Ebd}
\mathcal{E}_{\mathrm{bd}}(T_1)
\triangleq
\bigcap_{t=1}^{T_1}\big\{|\tilde\varepsilon_t|\le B_t\big\}.
\end{equation}
By \Cref{eq:app_stage1_subexp_tail} and a union bound over $t=1,\dots,T_1$, the choice \Cref{eq:app_Bt_def}
ensures
\begin{equation}\label{eq:app_Ebd_prob}
\Pr\!\big(\mathcal{E}_{\mathrm{bd}}(T_1)\big)\ge 1-\delta_{\mathrm{bd}}.
\end{equation}
On $\mathcal{E}_{\mathrm{bd}}(T_1)$, the Stage~I feedback obeys
\begin{equation}\label{eq:app_bd_model}
\tilde z_t = f(i_t)+\tilde\varepsilon_t,
\;
|\tilde\varepsilon_t|\le B_t\le B_{T_1},
\ \ t=1,\dots,T_1.
\end{equation}
Since bounded random variables are sub-Gaussian, Hoeffding's lemma implies that on
$\mathcal{E}_{\mathrm{bd}}(T_1)$ the noise $\tilde\varepsilon_t$ is conditionally sub-Gaussian with
a conservative variance proxy on the order of $B_{T_1}^2$. Accordingly, in the GP posterior update
\Cref{eq:Ktilde} we may use an effective parameter $\sigma_\varepsilon^2 \gtrsim B_{T_1}^2$ for analysis.

\section{GP-LSE Guarantee}
\label{app:lse_overhead}

We provide a self-contained proof sketch of the standard GP-LSE guarantee on the finite domain $\mathcal I$.
Throughout, $(\mu_{t-1}(\cdot),\sigma_{t-1}(\cdot))$ denote the GP posterior mean and standard deviation
computed via kernel ridge regression with regularization $\sigma_\varepsilon^2$ in \Cref{eq:Ktilde}.
All results are stated on the bounded-noise event $\mathcal{E}_{\mathrm{bd}}(T_1)$ in \Cref{eq:app_Ebd},
under which $\tilde z_t=f(i_t)+\tilde\varepsilon_t$ and $|\tilde\varepsilon_t|\le B_{T_1}$ for all $t\le T_1$
and thus the noise is sub-Gaussian.

\subsection{Uniform confidence event}
\label{app:uniform_ci_thm6}

We work on the bounded-noise event $\mathcal E_{\mathrm{bd}}(T_1)$ constructed in
Appendix~\ref{app:stat_stage1}. On $\mathcal E_{\mathrm{bd}}(T_1)$, the debiased feedback obeys
\begin{equation}\label{eq:app_bd_model_recall}
\tilde z_t = f(i_t) + \tilde\varepsilon_t,\;
\mathbb E[\tilde\varepsilon_t \mid \mathcal F_{t-1}] = 0,\;
|\tilde\varepsilon_t| \le B_t,\ \forall t\le T_1,
\end{equation}
where $\{\mathcal F_t\}$ is the natural filtration and $\{B_t\}$ is defined in \Cref{eq:app_Bt_def}.
Since $\{B_t\}$ is nondecreasing by construction, it follows that
$|\tilde\varepsilon_t|\le B_{T_1}$ for all $t\le T_1$.

We further assume the standard RKHS complexity condition used by kernelized bandit analyses:
the unknown power map $f$ is a fixed function in $\mathcal H_k$ with $\|f\|_{\mathcal H_k}\le B_f$. This is an agnostic function-class condition: it requires that $f$ be regular under $k$.
Because our cross-pattern kernel \Cref{eq:plus_laplace_kernel} encodes the near-field DFT-beam locality observed in
\Cref{fig:upa_mainlobe}-\Cref{fig:upa_main_lobe2}, it is reasonable that the induced RKHS complexity of $f$ is moderate.

Recall the GP confidence interval in the main text:
{\small\begin{equation}\label{eq:app_Qt_def_recall}
Q_t(i)
\triangleq
\Big[\mu_{t-1}(i)-\sqrt{\beta_t}\sigma_{t-1}(i),\;
     \mu_{t-1}(i)+\sqrt{\beta_t}\sigma_{t-1}(i)\Big],
\end{equation}}
which matches \Cref{eq:lse_Q}. The shrinking interval used by GP-LSE is defined as
{\small\begin{equation}\label{eq:app_Ct_def}
C_t(i)\triangleq \bigcap_{s=1}^t Q_s(i),\;
\ell_t(i)=\min C_t(i),\; u_t(i)=\max C_t(i).
\end{equation}}

\begin{lemma}
\label{lem:app_uniform_ci_thm6}
Fix $\delta\in(0,1)$. Choose a nondecreasing exploration sequence
\begin{equation}\label{eq:app_beta_thm6}
\beta_t \;=\; 2B_f^2 \;+\; 300\,\gamma_{t-1}\,\log^3\!\Big(\tfrac{t}{\delta}\Big),
\; t\ge 1,
\end{equation}
where $\gamma_{t}$ is the maximum information gain of kernel $k$ on $\mathcal I$ under
noise parameter $\sigma_\varepsilon^2$.
Then, conditioned on $\mathcal E_{\mathrm{bd}}(T_1)$, with probability at least $1-\delta$,
for all $t\ge 1$ and all $i\in\mathcal I$,
\begin{equation}\label{eq:app_eventE_thm6}
\big|f(i)-\mu_{t-1}(i)\big|
\le
\sqrt{\beta_t}\,\sigma_{t-1}(i).
\end{equation}
Consequently, $f(i)\in Q_t(i)$ for all $i,t$, and thus $f(i)\in C_t(i)$ for all $i,t$.
\end{lemma}

\begin{proof}
On $\mathcal E_{\mathrm{bd}}(T_1)$, the observation noise $\tilde\varepsilon_t$ in
\Cref{eq:app_bd_model_recall} is a uniformly bounded martingale difference sequence over $t\le T_1$.
Therefore, the RKHS-based GP bandit confidence analysis applies. In particular, the escape-event argument
and the martingale concentration step in Theorem~6 of \cite{srinivas2010gpucb} imply that the posterior
deviation is uniformly controlled by $\sqrt{\beta_t}\sigma_{t-1}(i)$ when $\beta_t$ is chosen as in
\Cref{eq:app_beta_thm6}. This yields \Cref{eq:app_eventE_thm6}. The last inclusion follows from the
definitions \Cref{eq:app_Qt_def_recall}-\Cref{eq:app_Ct_def}.
\end{proof}

Define the overall “good” event as
\begin{equation}\label{eq:app_good_event}
\mathcal E \triangleq \mathcal E_{\mathrm{bd}}(T_1)\cap \mathcal E_{\mathrm{ci}},
\end{equation}
where $\mathcal E_{\mathrm{ci}}$ denotes the uniform confidence event in \Cref{eq:app_eventE_thm6}.
Since $\Pr(\mathcal E_{\mathrm{bd}}(T_1))\ge 1-\delta_{\mathrm{bd}}$ and
$\Pr(\mathcal E_{\mathrm{ci}}\mid \mathcal E_{\mathrm{bd}}(T_1))\ge 1-\delta$,
\begin{equation}\label{eq:app_good_event_prob}
\begin{aligned}
\Pr(\mathcal E)
&=\Pr(\mathcal E_{\mathrm{bd}}(T_1)\cap \mathcal E_{\mathrm{ci}}) \\
&=\Pr(\mathcal E_{\mathrm{bd}}(T_1))\,\Pr(\mathcal E_{\mathrm{ci}}\mid \mathcal E_{\mathrm{bd}}(T_1)) \\
&\ge (1-\delta_{\mathrm{bd}})(1-\delta)
\ge 1-(\delta_{\mathrm{bd}}+\delta).
\end{aligned}
\end{equation}

\subsection{Correctness of the level-set estimate}

We now prove the $\epsilon$-accurate inclusion property under the event $\mathcal{E}$ in
\Cref{eq:app_good_event}. Recall the classification rule in \Cref{eq:lse_H}-\Cref{eq:lse_U} and the fact
$f(i)\in C_T(i)$ from Lemma~\ref{lem:app_uniform_ci_thm6}.

\paragraph{Upper inclusion.}
Fix any $i\in\mathcal{H}_T$. By the classification rule \Cref{eq:lse_H} we have
$\ell_T(i)>\tau-\epsilon$. On the event $\mathcal{E}$, Lemma~\ref{lem:app_uniform_ci_thm6} gives
$f(i)\in C_T(i)$, hence $f(i)\ge \min C_T(i)=\ell_T(i)$.
Combining the two inequalities yields $f(i)>\tau-\epsilon$, i.e.,
$i\in\{j\in\mathcal{I}: f(j)\ge \tau-\epsilon\}$.
Since $i\in\mathcal{H}_T$ was arbitrary, we conclude
\begin{equation}\label{eq:app_upper_inclusion}
\mathcal{H}_T\subseteq \{i\in\mathcal{I}: f(i)\ge \tau-\epsilon\}.
\end{equation}

\paragraph{Lower inclusion.}
Let $i$ satisfy $f(i)\ge \tau+\epsilon$. On $\mathcal{E}$, Lemma~\ref{lem:app_uniform_ci_thm6} implies
$f(i)\in C_t(i)$ for all $t$. In particular, $u_t(i)\ge f(i)\ge \tau+\epsilon$ for all $t$, so $i$ can
never satisfy the (lower) classification condition for $\mathcal{L}_t$ in \Cref{eq:lse_L}. Hence $i$ is
never permanently placed into $\mathcal{L}_t$. When the algorithm terminates, it enforces
$\mathcal{U}_T=\emptyset$, so every index is classified into either $\mathcal{H}_T$ or $\mathcal{L}_T$.
Therefore such $i$ must belong to $\mathcal{H}_T$, i.e.,
\begin{equation}\label{eq:app_lower_inclusion}
\{i:\ f(i)\ge \tau+\epsilon\}\subseteq \mathcal{H}_T.
\end{equation}

Combining \Cref{eq:app_upper_inclusion} and \Cref{eq:app_lower_inclusion} gives the desired
$\epsilon$-accurate level-set guarantee.

\subsection{Sample complexity via information gain}
\label{app:sample_complexity_ig}

We bound the number of GP-LSE queries until termination. The key link is that GP posterior variances
accumulate at most on the order of the kernel information gain; this is standard in GP-LSE analyses
\cite{srinivas2010gpucb,gotovos2013lse}.

Recall $Q_t(i)$ in \Cref{eq:app_Qt_def_recall} and $C_t(i)=\cap_{s\le t}Q_s(i)$ in \Cref{eq:app_Ct_def}.
Since $C_t(i)\subseteq Q_t(i)$, its width is no larger than that of $Q_t(i)$.
Thus, for any $i\in\mathcal{U}_t$,
\begin{equation}\label{eq:app_ambiguity_sigma}
\begin{aligned}
a_t(i)
&\le \tfrac{1}{2}\big(u_t(i)-\ell_t(i)\big)\\
&\le \tfrac{1}{2}\big(\max Q_t(i)-\min Q_t(i)\big)\\
&= \sqrt{\beta_t}\,\sigma_{t-1}(i).
\end{aligned}
\end{equation}


\begin{lemma}
\label{lem:app_var_ig}
Let $\{i_s\}_{s=1}^t$ be the queried sequence and let $\sigma_{s-1}^2(i_s)$ be the GP posterior variance
computed with regularization $\sigma_\varepsilon^2$ as in \Cref{eq:Ktilde}. Then (refer equation; sample complexity)
\begin{equation}\label{eq:app_logdet_decomp}
\log\det\!\Big(\mathbf{I}_t+\sigma_\varepsilon^{-2}\mathbf{K}_t\Big)
= \sum_{s=1}^t \log\!\Big(1+\sigma_\varepsilon^{-2}\sigma_{s-1}^2(i_s)\Big),
\end{equation}
where $\mathbf{K}_t$ is kernel matrix defined in \Cref{sec:stage1_gpucb}. Consequently,
\begin{equation}\label{eq:app_sum_sigma2}
\sum_{s=1}^t \sigma_{s-1}^2(i_s)
\le \frac{2}{\log(1+\sigma_\varepsilon^{-2})}\,\gamma_t
\triangleq C_1\,\gamma_t,
\end{equation}
where $\gamma_t$ is defined in \Cref{eq:app_gamma_def} and $C_1>0$ depends only on $\sigma_\varepsilon^2$.
\end{lemma}

\begin{proof}
The details can be found in \cite{srinivas2010gpucb} and \cite{gotovos2013lse} and here we provide a proof sketch. The identity \Cref{eq:app_logdet_decomp} is the standard incremental log-determinant decomposition for GP
posteriors. Since $\sigma_{s-1}^2(i_s)\in[0,k(i_s,i_s)]$ and 
$k(i,i)\le 1$, concavity of $\log(1+\sigma_\varepsilon^{-2}x)$ on $x\in[0,1]$ yields
$\log(1+\sigma_\varepsilon^{-2}x)\ge x\log(1+\sigma_\varepsilon^{-2})$.
Summing this inequality over $s$ and combining with \Cref{eq:app_logdet_decomp} gives
$\sum_{s=1}^t \sigma_{s-1}^2(i_s)\le \log\det(\mathbf{I}_t+\sigma_\varepsilon^{-2}\mathbf{K}_t)/\log(1+\sigma_\varepsilon^{-2})$.
Finally, by definition of $\gamma_t$ in \Cref{eq:app_gamma_def},
$\log\det(\mathbf{I}_t+\sigma_\varepsilon^{-2}\mathbf{K}_t)\le 2\gamma_t$.
\end{proof}

Let $m_t\triangleq \max_{i\in\mathcal{U}_{t-1}} a_t(i)$. By the GP-LSE selection rule \Cref{eq:lse_select},
$m_t=a_t(i_t)$. Moreover, $\{m_t\}$ is non-increasing because confidence sets shrink and the undecided set
only shrinks \cite{gotovos2013lse}. Hence $t\,m_t^2 \le \sum_{s=1}^t m_s^2=\sum_{s=1}^t a_s^2(i_s)$.

Using \Cref{eq:app_ambiguity_sigma} and monotonicity of $\{\beta_t\}$,
\begin{align}
\sum_{s=1}^t a_s^2(i_s)
&\le \sum_{s=1}^t \beta_s\,\sigma_{s-1}^2(i_s)
\le \beta_t \sum_{s=1}^t \sigma_{s-1}^2(i_s)
\le C_1\,\beta_t\,\gamma_t,
\end{align}
where the last step uses Lemma~\ref{lem:app_var_ig}. Therefore,
\begin{equation}\label{eq:app_amb_decay}
m_t \le \sqrt{\frac{C_1\,\beta_t\,\gamma_t}{t}}.
\end{equation}
A sufficient condition for the stopping criterion $m_t\le\epsilon$ is
\begin{equation}\label{eq:app_stop_cond}
t \ge \frac{C_1\,\beta_t\,\gamma_t}{\epsilon^2}.
\end{equation}
With $\beta_t$ chosen as in Theorem~6 of \cite{srinivas2010gpucb}, i.e., \Cref{eq:app_beta_thm6},
we obtain the stated Stage~I sample complexity
\begin{equation}
T_1
= \tilde{\mathcal{O}}\!\Big(\tfrac{\beta_{T_1}\gamma_{T_1}}{\epsilon^2}\Big)
= \tilde{\mathcal{O}}\!\Big(\tfrac{B_f^2\gamma_{T_1}+\gamma_{T_1}^2}{\epsilon^2}\Big).
\end{equation}

\section{Statistical Rationale of the Rician Denoising in Stage~II}
\label{app:stat_stage2}

This appendix explains why the Stage~II pseudo-amplitude
$\tilde{\psi}_p=\sqrt{[y_p^2-\sigma^2]_+}$ in \Cref{eq:rician_denoise}
can be interpreted as a noisy observation of the noiseless amplitude, thereby
connecting Stage~II to the standard noisy-amplitude sparse phase retrieval
formulation \cite{wang2018sparta_tsp}.

In Stage~II, the low-dimensional magnitude measurement is
$y_p=\big|\mathbf{g}_p^H\mathbf{s}_{\widehat{\mathcal{S}}}+w_p\big|$ with
$w_p\sim\mathcal{CN}(0,\sigma^2)$. Let
\begin{equation}
a_p \triangleq \big|\mathbf{g}_p^H\mathbf{s}_{\widehat{\mathcal{S}}}\big|
\end{equation}
denote the noiseless amplitude. Consider the debiased intensity
\begin{equation}
\tilde{z}_p \triangleq y_p^2-\sigma^2 .
\end{equation}
Using the same decomposition as in \Cref{app:stat_stage1} (with
$s\leftarrow a_p e^{j\varphi_p}$), we obtain the conditional unbiasedness
\begin{equation}
\tilde{z}_p
= a_p^2 + \varepsilon_p,
\;
\mathbb{E}\!\left[\varepsilon_p\,\big|\,\mathbf{g}_p\right]=0,
\label{eq:app_stage2_intensity}
\end{equation}
and the conditional variance
\begin{equation}
\mathrm{Var}\!\left(\varepsilon_p\,\big|\,\mathbf{g}_p\right)
=\sigma^4+2\sigma^2 a_p^2 .
\label{eq:app_stage2_var_eps}
\end{equation}
Hence, $\tilde{z}_p$ is a zero-mean-fluctuated observation of $a_p^2$ with
heteroscedastic variability.

Stage~II applies a nonnegativity projection and square-root mapping:
\begin{equation}
\tilde{\psi}_p
= \sqrt{\left[y_p^2-\sigma^2\right]_+}
= \sqrt{\left[a_p^2+\varepsilon_p\right]_+}.
\label{eq:app_stage2_psitilde}
\end{equation}
When the SNR is moderate-to-high so that $a_p$ is not too small for most $p$,
we can relate $\tilde{\psi}_p$ to $a_p$ through a deterministic bound.
Specifically, on the event $a_p^2+\varepsilon_p\ge 0$ (which is always enforced
by $[\cdot]_+$), by the mean value theorem applied to $\sqrt{x}$, we have
\begin{equation}
\big|\tilde{\psi}_p-a_p\big|
=\big|\sqrt{[a_p^2+\varepsilon_p]_+}-\sqrt{a_p^2}\big|
\le
\frac{|\varepsilon_p|}{\sqrt{[a_p^2+\varepsilon_p]_+}+a_p}.
\label{eq:app_stage2_mvt_bound}
\end{equation}
In particular, whenever $|\varepsilon_p|\le a_p^2/2$ (hence
$[a_p^2+\varepsilon_p]_+\ge a_p^2/2$), the denominator in
\Cref{eq:app_stage2_mvt_bound} is lower bounded as
$\sqrt{[a_p^2+\varepsilon_p]_+}+a_p \ge (1+1/\sqrt{2})a_p$, yielding
\begin{equation}
\big|\tilde{\psi}_p-a_p\big|
\le
\frac{|\varepsilon_p|}{(1+1/\sqrt{2})a_p}
\le
\frac{|\varepsilon_p|}{a_p}.
\label{eq:app_stage2_simple_bound}
\end{equation}
Such a condition is mild for measurements probing dominant-support beams discovered in Stage~I,
since the corresponding $a_p$ is typically above the noise floor for a large fraction of $p$. Therefore, $\tilde{\psi}_p$ can be interpreted as a {noisy amplitude}
observation of $a_p$, where the effective disturbance is induced by the
Rician fluctuation $\varepsilon_p$ and the square-root mapping. 

From \Cref{eq:app_stage2_psitilde}-\Cref{eq:app_stage2_mvt_bound}, whenever
$\big[a_p^2+\varepsilon_p\big]_+>0$,
the mapping $\tilde{\psi}_p=\sqrt{\big[a_p^2+\varepsilon_p\big]_+}$ induces an
additive amplitude-noise model $\tilde{\psi}_p=a_p+\eta_p$ with a bounded
perturbation controlled by the bound in \Cref{eq:app_stage2_mvt_bound}.
Hence, conditioned on the event $\big[a_p^2+\varepsilon_p\big]_+>0$, the Stage~II
data fit exactly the “noisy amplitude” setting under which SPARTA is empirically
robust \cite{wang2018sparta_tsp}. On the other hand, when $a_p^2+\varepsilon_p\le 0$, the projection in
\Cref{eq:app_stage2_psitilde} clips the pseudo-amplitude to $\tilde{\psi}_p=0$,
which corresponds to below-noise-floor samples; these are naturally suppressed by
SPARTA’s truncation and hard-thresholding mechanism. Therefore, in the
moderate-to-high SNR regime, which is typical for received signal of dominant DFT codewords, the proposed amplitude construction is compatible with SPARTA and
leads to accurate refinement in practice.
\vspace{2mm}
\printbibliography

\end{document}